\documentclass[letterpaper]{scrartcl} 

\usepackage[]{brandl}
\usepackage{mathtools}
\usepackage{pgfplots} 
\usepackage{calc}
\usepackage{etoolbox}
\usepackage{bm,esvect,tabu}
\usepackage{tabularx}
\usepgfplotslibrary{colorbrewer,ternary,external}
\pgfplotsset{compat = newest}

\setkomafont{caption}{\footnotesize}
\setkomafont{captionlabel}{\normalsize} 

\theoremstyle{definition}

\DeclareMathOperator{\probability}{\mathbb P}% 
\DeclareMathOperator{\expectedvalue}{\mathbb E}% 
\DeclareMathOperator{\variance}{Var}% 
\newcommand{\ev}[1][]{% 
\ifthenelse{\equal{#1}{}}{{\expectedvalue #1}}{{\expectedvalue\left(#1\right)}}% 
}
\newcommand{\evcond}[2][]{% 
\ifthenelse{\equal{#1}{}}{{\expectedvalue_{#1} #2}}{{\expectedvalue_{#1}\left(#2\right)}}% 
}
\newcommand{\var}[1][]{% 
\ifthenelse{\equal{#1}{}}{{\variance #1}}{{\variance\left[#1\right]}}% 
}
\newcommand{\pr}[1][]{% 
\ifthenelse{\equal{#1}{}}{{\probability #1}}{{\probability\left(#1\right)}}% 
}
\newcommand{\prcond}[2][]{% 
\ifthenelse{\equal{#1}{}}{{\probability_{#1} #2}}{{\probability_{#1}\left(#2\right)}}% 
}

\DeclarePairedDelimiterX{\infdivx}[2]{(}{)}{% 
  #1\;\delimsize\|\;#2% 
}
\newcommand{\infdiv}{D\infdivx}

\newcommand{\raut}{r_{\mathrm{aut}}}
\newcommand{\rbdd}{r_{\mathrm{bdd}}}
\newcommand{\rbddt}{\tilde r_{\mathrm{bdd}}}
\newcommand{\rmajt}{\hat r_{\mathrm{bdd}}}

\newcolumntype{Y}{>{\centering\arraybackslash}X}

\title{The Social Learning Barrier}

\author{Florian Brandl\thanks{\texttt{florian.brandl@uni-bonn.de}}\\ University of Bonn}

\date{}

\begin{document}
\maketitle

\begin{abstract}
	We consider long-lived agents who interact repeatedly in a social network. In each period, each agent learns about an unknown state by observing a private signal and her neighbors' actions from the previous period before choosing her own action. Our main result shows that the learning rate of the slowest-learning agent is bounded from above by a constant that only depends on the marginal distributions of the agents' private signals and not on the number of agents, the network structure, correlations between the private signals, and the agents' strategies. Applying this result to equilibrium learning with rational agents shows that the learning rate of all agents in any equilibrium is bounded under general conditions. This extends recent findings on equilibrium learning and demonstrates that the limitation stems from an information-theoretic tradeoff between optimal action choices and information revelation, rather than strategic considerations. We also show that a social planner can achieve almost optimal learning by designing strategies for which each agent's learning rate is close to the upper bound.  
\end{abstract}

\section{Introduction}\label{sec:introduction}

How fast do individuals learn from repeatedly observing each other's actions in social networks? 
The amount of private information in large networks is vast, so efficient information aggregation would lead to rapid learning.
However, we show that information aggregation fails under general conditions: the rate of learning of the slowest-learning agent is bounded from above by a constant that does not depend on the size and structure of the network and the agents' behavior. 
This has various direct consequences for equilibrium learning of rational agents and non-Bayesian agents.
The general insights apply to domains such as product choice, voting, technology adoption, and opinion formation.
  
In our model, long-lived agents interact with one another over an infinite number of periods in a network. 
The state of the world is fixed but unknown.
In each period, each agent receives a private signal about the state and observes the actions of her neighbors from the previous period before choosing her own action.
The signals are independent and identically distributed across periods conditional on the state.
An agent's flow utility in a period depends on her action and the state, but not on the actions of other agents, and is unobserved.
We quantify learning by the rate at which the probability of suboptimal action choices vanishes.  
Each agent can learn the state in the limit through the constant stream of her private signals alone. 
However, our main result (\Cref{thm:main}) shows that information aggregation fails under general conditions: for any number of agents, any network structure, and any strategies for the agents, some agent learns no faster than a fixed upper bound.
This bound only depends on the marginal distributions of the agents' private signals and not on correlations between them.
Thus, even a social planner who can dictate the agents' strategies cannot achieve a uniformly high learning rate.

In networks with many agents, the total amount of private information is vast. 
So fully efficient information sharing would enable all agents to surpass any desired learning rate as the network grows in size.
However, we show that some agents inevitably learn not much faster than a single isolated agent, independently of network size.
To see what causes the breakdown of information aggregation, take the perspective of a social planner who can design the agents' strategies.\footnote{In online social networks, the platform can significantly influence the users' behavior, similar to our imaginary social planner, by selecting the information shown to a user and making personalized recommendations.}
The planner aims to maximize the rate at which the probability of suboptimal action choices vanishes, which corresponds to maximizing the learning rate of the slowest-learning agent in our model.
The strategies must trade off two competing objectives: on the one hand, for each agent, the probability that her action is suboptimal must vanish rapidly; on the other hand, an agent's action choice must contain information about her private signals.
To achieve a learning rate significantly higher than that of an isolated agent, the first objective requires that an agent's actions do not depend too often on her private signals.
But this conflicts with the second objective, which requires the opposite.
In fact, strategies that depend on the agents' private signals too infrequently result in information cascades similar to the ``rational herds'' of \citet*{BHW92a} or the ``rational groupthink'' event of \citet*{HMST21a}.

The obtained bound on the learning rate is tight: we show that a social planner can design strategies for which each agent's learning rate gets arbitrarily close to the upper bound if there are sufficiently many agents and the network is strongly connected (\Cref{thm:coordination}).\footnote{A network is strongly connected if there is an observational path from any agent to any other agent.}
Here, we assume that the signals are conditionally independent and identically distributed across periods and agents.
This shows that the above tradeoff is not absolute: agents can match the correct action more frequently than an isolated agent and have their actions be informative about their private signals at the same time.
For complete networks, the strategies we use are simple.
Each agent follows her past private signals if those are highly indicative of some state, which ensures that these actions are very likely correct.
Whenever an agent's private signals do not decisively favor one of the states, she copies the action that was most popular among all agents in the previous period.
Classic results from large deviations theory show that most agents' signals strongly indicate the true state in most periods, so that the most popular action in any period is very likely correct.
Hence, an agent's action is very likely correct in either case.

We then turn to learning in equilibrium by rational agents who geometrically discount future payoffs.
An imitation argument of \citet{HST24a} shows that all agents learn at the same rate in any equilibrium and any strongly connected network.\footnote{See \citet[][Lemma 2]{HST24a} for the imitation argument with geometrically discounting agents. Similar imitation principles are common in the literature \citep[see, e.g.,][]{SmSo00a,GaKa03a,GoSa17a}.\label{fn:imitation}}
Hence, any bound on the learning rate of the slowest-learning agent gives the same bound for every agent.
This allows us to recover the results of \citet{HMST21a} and \citet{HST24a}---showing that all agents learn at a bounded rate in any equilibrium for any number of agents in any strongly connected network---under more general conditions.
Most notably, we do not assume that the set of actions is finite, that the signals are conditionally independent across agents, or that each signal has bounded informativeness.
The main new conceptual insight is that the bound on the learning rate arises not from equilibrium strategy restrictions, but from an information-theoretic tradeoff.
Moreover, since our bound on the learning rate of the slowest-learning agent applies to any strategies, it holds for equilibrium learning independently of agents' evaluation of future payoffs, for misspecified agents, and for agents who use non-Bayesian heuristics.
Whenever an imitation argument is available or the strategy profile is symmetric, the bound holds for all agents, not just the slowest-learning one.

Finally, we demonstrate that the failure of information aggregation persists in a richer information environment by considering a model variation where agents observe their neighbors' signals in addition to their actions.
In complete networks, all information is then public, and rational agents learn at the first-best rate.
However, a common feature of social networks is that each agent has a small number of neighbors compared to the size of the network.
As a corollary of our main result, we show that the size of the largest neighborhood controls the equilibrium learning rate.
That is, in any equilibrium for any number of agents in any strongly connected network, no agent's learning rate exceeds an upper bound that only depends on the size of the largest neighborhood and the marginal distributions of private signals.
For this result, we assume that the signals are conditionally independent across periods and agents, and agents are rational and geometrically discount future payoffs.

These results lead to four main insights.
First, observational learning in networks fails to efficiently aggregate information, even when a social planner designs the agents' strategies.
Second, a small amount of information aggregation is feasible: suitably designed strategies allow each agent to learn faster than a single, isolated agent.
Third, upper bounds on the equilibrium learning rate are not primarily driven by incentives.
And fourth, these bounds hold even when agents also observe their neighbors' signals in networks with small neighborhoods.

The rest of the paper is structured as follows. 
\Cref{sec:related-work} discusses related work and \Cref{sec:model} introduces the model.
In \Cref{sec:autarky-equilibrium}, we recall known results on learning by an isolated agent and with publicly observable signals.
\Cref{sec:coordinated-learning} states the main result, explains the ideas underlying its proof, and shows that the established bound on the learning rate is tight.
\Cref{sec:equilibrium-learning} applies this result to learning in equilibrium with geometrically discounting agents and considers the model variation where agents observe their neighbors' signals and actions. 
\Cref{sec:proof-outlines} outlines the proofs of the main results.
\Cref{sec:discussion} concludes with a discussion of modeling assumptions and future directions.
All proofs are in the \hyperref[sec:app-prelims]{Appendix}.

\section{Related Work}\label{sec:related-work}

Most of the literature has focused on equilibrium learning, non-Bayesian agents, non-recurring private signals, or sequentially arriving agents.

\subsection{Equilibrium Learning in Networks}

Studying models with multiple periods and long-lived rational agents is challenging.
Agents may choose suboptimal actions today to induce other agents to reveal information tomorrow.
In recent work, \citet{HST24a} show that in the same model as in the present paper, the rate of learning of any agent in any equilibrium is bounded independently of the number of agents in strongly connected networks. 
This result follows from the following elegant argument: if an agent's neighbor learns much faster than in autarky, she will ignore her private signals and imitate her neighbor, so that her actions cease to reveal information about her private signals; in a strongly connected network, this reasoning extends inductively to all agents; but then all agents eventually ignore their private signals, learning stops, and a contradiction is reached.

\citet{HMST21a} obtain a similar conclusion in a restricted setting with myopic agents and networks in which each agent observes all other agents' actions.
They use large deviations theory to derive their result, and their work is methodologically closer to ours.
Both papers rely on agents being fully rational, playing equilibrium strategies, and either geometrically discounting future payoffs or being myopic---assumptions that our results show are not necessary.

\subsection{Non-Bayesian Learning}

Because of the difficulties arising from Bayesian learning with repeated interactions, the literature has focused on learning heuristics and non-Bayesian agents.
The literature following \citet{Degr74a} assumes that agents observe each other's beliefs and form tomorrow's belief via an updating heuristic such as linear aggregation \citep[see, e.g.,][]{GoJa10a}.
Another approach is to relax the assumptions that agents are fully Bayesian. 
For example, \citet{BaGo98a} assume agents respond rationally to private signals and others' random payoffs but ignore the informational content of others' actions. Meanwhile, the agents in \citet*{MTJ18a} use a heuristic: they combine past beliefs and then update that aggregate belief rationally based on their private signals.

\subsection{Non-Recurring Private Signals} 

Another strand of the literature, starting with \citet{GePo82a}, \citet{Bach85a}, and \citet{PaKr90a}, considers models in which rational and long-lived agents receive a private signal once in the first period and repeatedly observe the actions of other agents, and studies whether agents converge on the same action. 
\citet{GaKa03a} allow for social networks in which agents observe their neighbors' actions and show that eventually, all agents converge on the same action.
In the same model, \citet*{MST14a,MST15a} study the probability that agents converge on the correct action as the number of agents goes to infinity and show that it depends on the network structure.
\citet{Vive93a} considers a continuum of agents with a continuous action space and shows that information aggregation can still be slow if observations of actions are noisy, as in the case of observed market prices.
In contrast to our model, agents do not receive private signals in later periods. 

\subsection{Sequentially Arriving Agents}

In the classical herding model \citep*{BHW92a,Bane92a,SmSo00a}, agents arrive sequentially, observe a private signal as well as their predecessors' actions, and take an action once.  
Learning can fail in this setting since rational agents may ignore their private signals and follow their predecessors' actions, leading to herding on the wrong action. 
The analysis of this model is substantially different from the present model. 
First, since each agent acts only once, informational feedback loops need not be considered.
Second, each agent receives only one private signal, so that rational agents may fail to learn the state in the limit.

\citet{RoVi17a}, in an unpublished draft of a later paper \citep{RoVi19a}, consider a social planner who can dictate the agents' strategies.
They demonstrate that the planner can achieve efficient learning (in the sense that the expected number of agents failing to match the state is finite), even if the distribution of private signals is not informative enough for efficient learning in an equilibrium with rational agents.
Hence, the inefficiency of equilibrium learning arises from strategic incentives rather than a more basic, information-theoretic tradeoff, which contrasts with our results. 
Relatedly, \citet*{ABM+25a} consider condescending agents who underestimate the quality of the other agents' private signals.
This misspecification decreases the probability of herding on the wrong action and improves learning compared to the well-specified case if condescension is mild.
Thus, they recover the result of \citet{RoVi17a} since a social planner could instruct agents to behave as if they were condescending.
\citet{HMST21a} conjecture that the same misspecification improves learning in the present model as well.
Our results show that its benefit can be small at best.
\citet{SST21a} study a social planner who aims to maximize the discounted social welfare.
These works are similar to our approach in spirit.

In a variant of the herding model, the state changes stochastically over time \citep*{MOS98a,LPV24a,Huan24a}.
We maintain the assumption that the state is persistent throughout.
A recent survey of \citet*{BHTW21a} summarizes the work on models with sequentially arriving agents.

\section{The Model}\label{sec:model}

Let $N = \{1,\dots,n\}$ be the set of agents and let $T = \{1,2,\dots\}$ be the set of periods.
Each agent has the same possibly infinite set of actions $A$ and chooses an action in each period.
If $x$ is a vector indexed by $N$ and $i\in N$, then $x^i$ denotes its $i$-th coordinate, and if $x$ is indexed by $T$ and $t\in T$, then $x_t$ is the $t$-th coordinate, $x_{\le t}$ is its restriction to the periods $\{1,\dots,t\}$, and $x_{<t}$ is its restriction to the periods $\{1,\dots,t-1\}$.
There is an unknown state $\omega$, taking values in a finite set $\Omega$.
Generic elements of $\Omega$ are denoted by $f,g$.
We denote by $\pi_0\in\Delta(\Omega)$ the distribution of $\omega$, and assume that it has full support.
For a state $f\in \Omega$, we write $\prcond[f]{\cdot} = \pr[\cdot \mid f]$ and $\evcond[f]{\cdot} = \ev[\cdot\mid f]$ for the corresponding conditional probability and conditional expectation.

\subsection{Agents' Payoffs}

All agents have the same utility function $u\colon A\times\Omega\rightarrow\mathbb R$ that depends on their own action and the state, and $u(a,\omega)$ is an agent's flow utility for choosing the action $a$ in any period.
An agent's utility is independent of other agents' actions, so the interactions between the agents are purely informational.
We assume that the optimal action $a_f$ for every state $f\in\Omega$ exists and is unique.
\begin{align*}
	\{a_f\} = \arg\max_{a\in A} u(a,f)
\end{align*}
We also assume that $a_f\neq a_{g}$ for any two distinct states $f,g\in\Omega$.\footnote{The assumption that all agents have the same utility function is purely for notational convenience. All proofs remain valid with the obvious adjustments provided each agent's utility function satisfies the preceding genericity assumptions and the agents know each other's utility functions (or at least each other's optimal actions in each state).}
We say that $a_\omega$ is the correct action and any other action is a mistake.
The requirement that no action is optimal in two different states avoids trivial cases, and the uniqueness of the optimal action for each state precludes agents from encoding information in the selection of the optimal action.
This tradeoff between choosing the correct action and communicating information about private signals is the main tension in our model.
We explain this in detail in \Cref{rem:non-unique-optimal-actions}.
On the other hand, a single action choice can contain an arbitrary amount of information since $A$ can be infinite, which makes our result establishing an upper bound on the learning rate stronger.
Since we quantify learning by the probability of choosing the correct action, the sole role of utility functions is to identify the correct action.
More fine-grained characteristics of the utility functions only become relevant when analyzing equilibria with geometrically discounting agents (cf.\ \Cref{sec:equilibrium-learning}).

\subsection{Agents' Information}

The prior distribution $\pi_0$ is commonly known.
In each period $t\in T$, each agent $i$ privately observes a signal $\mathfrak s_t^i$ from a set of signals $S$, which is assumed to be a standard Borel space.
Conditional on each state $f$, $\mathfrak s_t^i$ has distribution $\mu_f^i\in\Delta(S)$ and $\mathfrak s_t = (\mathfrak s_t^i)_{i\in N}$ has distribution $\mu_f\in\Delta(S^N)$ for each $t\in T$, and $(\mathfrak s_t)_{t\in T}$ are independent conditional on $\omega$.
That is, signals are conditionally independent across periods but not necessarily across agents.
We assume that $\mu_f,\mu_{g}$ are mutually absolutely continuous for any two states $f,g$, so that no signal excludes any state with certainty.
Observing the signal realization $s\in S$ changes the log-likelihood ratio of the observing agent $i$ between the states $f$ and $g$ by 
\begin{align*}
	\ell^i_{f,g}(s) = \log \frac{d\mu^i_f}{d\mu^i_{g}}(s)
\end{align*}
Let $\ell^i_{f,g} = \ell^i_{f,g}(\mathfrak s_1^i)$ for any pair of states $f,g$ and any agent $i$.
We assume that $\ell^i_{f,g}$ is not identically zero and that its cumulant generating function $\lambda_{f,g}^i(z) = \log \ev[e^{z \ell^i_{f,g}}]$ is finite for each $z\in\mathbb R$ and any two distinct states $f,g$. 
This assumption allows us to liberally use results from large deviations theory, and is satisfied, for example, whenever $\ell_{f,g}^i$ is bounded or if $S = \mathbb R$ and $\mu_f^i$ is a normal distribution with the same variance for each state $f$.

The private signals of agent $i$ up to any period~$t$ induce the private log-likelihood ratio 
\begin{align*}
	L_{f,g,t}^{i} = \log \frac{\pi_0(f)}{\pi_0(g)} + \sum_{r \le t} \ell^i_{f,g}(\mathfrak s_r^i)
\end{align*}

Each agent $i$ observes the actions of her neighbors $N^i\subset N$, and we assume $i\in N^i$ so that each agent observes her own action.
The directed graph induced by these neighborhoods is called the network, and we assume that it is common knowledge among the agents.
A network is strongly connected if there is an observational path from any agent to any other agent, and complete if each agent's neighborhood is $N$.

Agents do not observe each other's signals.
Moreover, agents know their utility function (and thus everyone's utility function) but do not observe the flow utility $u(a,\omega)$ of any agent, including themselves. 
The latter assumption shuts down experimentation motives and is common for models of learning without experimentation.
Our formulation includes a model in which agents receive noisy signals about their flow utility today through tomorrow's signal.\footnote{Formally, consider the case that agent $i$'s flow utility for action $a$ in period~$t$ is $\tilde u(a,\mathfrak s_{t+1}^i)$ for an action and signal-dependent utility function $\tilde u\colon A\times S\rightarrow\mathbb R$.
Agents thus observe their flow utility in period~$t$ through their signal in the next period.
If we define $u(a,f) = \evcond[f]{\tilde u(a,\mathfrak s_1^i)}$, then both models are equivalent in terms of expected payoffs at the time of choosing an action.
This connection has also been noted by \citet*{RSV09a}, \citet{HMST21a}, and \citet{HST24a}.
\label{fn:experimentation}}
In \Cref{rem:experimentation}, we explain that the assumption of unobserved utilities can be weakened substantially.

Thus, the information available to agent $i$ in period~$t$ before choosing an action consists of the actions of all of $i$'s neighbors in all previous periods and $i$'s signals in all periods up to and including $t$.
We say that $A^{|N^i|\times (t-1)}$ is the set of public histories of agent $i$, $S^t$ is the set of private histories for each agent, and $\mathcal I_{\le t}^i = S^t \times A^{N^i\times (t-1)}$ is the collection of information sets of agent $i$ before choosing an action in period~$t$.

\subsection{Agents' Strategies}\label{sec:strategies}

A pure strategy for agent $i$ is a sequence $\sigma^i = (\sigma_t^i)_{t\in T}$ of measurable functions $\sigma_t^i\colon \mathcal I_{\le t}^i\rightarrow A$ from $i$'s information sets in period $t$ to the set of actions, and a pure strategy profile $\sigma = (\sigma^1,\dots,\sigma^n)$ consists of a strategy for each agent.
A pure strategy profile $\sigma$ induces a random sequence of action profiles: for each $i$, $a_1^i = \sigma_1^i(\mathfrak s_1^i)$, and for each $t > 1$, $a_t^i = \sigma_t^i(\mathfrak s_{\le t}^i;(a_{<t}^j)_{j\in N^i})$, where $i$'s public history $(a_{<t}^j)_{j\in N^i}$ is given by the actions of her neighbors in the periods preceding $t$.  
For $i\in N$, $t\in T$, and $H_{\le t}\in A^{N\times t}$, we denote by $\mathcal S^i(H_{\le t}) = \{s^i_{\le t} \in S^{t}\colon  \forall r \le t, \sigma_r^i(s^i_{\le r}; (H_{<r}^j)_{j\in N^i}) = H_r^i\}$ the set of those trajectories of $i$'s signals  consistent with $H_{\le t}$, and write $\mathcal S(H_{\le t}) = \prod_{i\in N} \mathcal S^i(H_{\le t})$ for the trajectories of signal profiles consistent with $H_{\le t}$.
Thus, play follows $H_{\le t}$ if and only if each agent $i$ receives signals in $\mathcal S^i(H_{\le t})$.

We say that agent $i$ makes a mistake in period~$t$ if $a_t^i\neq a_\omega$, and that $i$ learns at rate $r$ if
\begin{align*}
	r = \liminf_{t\rightarrow\infty} -\frac1t\log \pr[a_t^i\neq a_\omega]
\end{align*}
If the limit exists, the probability of a mistake in period~$t$ is $e^{-rt + o(t)}$.\footnote{We use the asymptotic notation $o(t)$ for a function that grows slower than $t$ as $t$ goes to infinity. That is, $\phi(t) \in o(t)$ if $\lim_{t\rightarrow \infty} \frac{\phi(t)}{t} = 0$.}
This definition of the learning rate is common in the literature \citep[see, e.g.,][]{Vive93a,HMT18a,MTJ18a,RoVi19a,HMST21a,HST24a}.\footnote{An alternative, more quantitative asymptotic definition of the learning rate considers the rate at which the expected difference between the utility of the correct action and an agent's action goes to $0$, i.e.,
	\begin{align*}
		\liminf_{t\rightarrow\infty} -\frac1t \log \ev[u(a_\omega,\omega) - u(a_t^i,\omega)]
	\end{align*}
	For finite action sets, both definitions coincide. 
	Indeed, $\liminf_{t\rightarrow\infty} -\frac1t \log \ev[u(a_\omega,\omega) - u(a_t^i,\omega)] = \liminf_{t\rightarrow\infty} -\frac1t \log c_t\pr[a_t^i\neq a_\omega] = \liminf_{t\rightarrow\infty} -\frac1t \log \pr[a_t^i\neq a_\omega]$ for some $c_t$ between the minimum and the maximum of the set $\{u(a_f,f) - u(a,f)\colon f\in \Omega, a \in A\setminus\{a_f\}\}$ for each $t\in T$.
	For infinite action sets and a bounded utility function, the learning rate is weakly higher for the quantitative definition and may be arbitrarily high even when the qualitative learning rate is $0$. 
	Hence, the qualitative definition of the learning rate is most reasonable for finitely many actions.}
For any fixed number of agents, the lowest rate of any agent determines how fast the expected number of mistakes per period vanishes, since the exponential with the smallest negative exponent dominates in the long run.
Thus, the rate of the slowest-learning agent measures the rate of learning in the network.
Moreover, any bound on the rate of learning for pure strategies entails the same bound for mixed strategies for the same instance with a larger signal space and an additional signal component that is uninformative about the state.% 
\footnote{More precisely, replace the signal space $S$ by $\tilde S = S \times [0,1]$, and for each state $f$, let $\tilde\mu_f$ be the product distribution on $\tilde S^N$ with marginal $\mu_f$ with respect to $S^N$ and the uniform distribution on $[0,1]^N$ as its marginal with respect to $[0,1]^N$.
We may choose the signals' second coordinates so that they are independent of the state and any other signals.
Then, the informativeness of each agent's signals remains unchanged, and the signal profile distributions remain mutually absolutely continuous and conditionally independent across periods.
The second coordinate of a signal in $\tilde S$ can be used to map any mixed strategy $\sigma$ for the instance with signals in $S$ to a pure strategy $\tilde\sigma$ for signals in $\tilde S$ that is behaviorally equivalent, i.e., conditional on each state, the induced distributions of sequences of action profiles (i.e., distributions on $A^{N\times T}$) are the same for $\sigma$ and $\tilde\sigma$.
Likewise, for each pure strategy profile with signal space $\tilde S$, there is a behaviorally equivalent mixed strategy profile with signal space $S$. \label{fn:pure-strategies}}
Hence, restricting to pure strategies entails no loss in generality, and we do so throughout.

All notation is summarized in \Cref{sec:referencetable}.

\subsection{Leading Example}\label{sec:leading-example}

The following simple instance of the model already presents most of the arising complexities and can serve as a leading example.
There are two states and two actions, say, $\Omega = \{f,g\}$ and $A = \{a_f,a_g\}$.
Each agent has utility $1$ for matching the state and $0$ for failing to match the state.
The network is complete.
The signals are binary ($S = \{s_f,s_g\}$) and conditionally independent and identically distributed across periods and agents, and each agent's signal in each period matches the state with probability $p$ for some $p\in(\frac12,1)$ (i.e., $\mu^i_f(s_f) = p$ and $\mu^i_g(s_g) = p$).
We illustrate our results using this example at the end of \Cref{sec:coordinated-learning}.

\section{Autarky and Public Signals}\label{sec:autarky-equilibrium}

We revisit two settings as benchmarks: a single agent learning in autarky and several agents observing each other's signals.

If there is only a single agent $i$ who chooses actions optimally based on her private signals, it follows from classical large deviations results for random walks that for some $\raut^i > 0$ determined by the signal distributions $(\mu_f^i)_{f\in \Omega}$,
\begin{align*}
	\pr[a_t^i\neq a_\omega] = e^{-\raut^i t + o(t)}
\end{align*} 
Hence, the limit $\lim_{t\rightarrow\infty} - \frac1t \log \pr[a_t^i\neq a_\omega]$ exists and is equal to $\raut^i$, which we call $i$'s autarky learning rate.
In particular, the probability of a mistake goes to zero as time goes to infinity.
For a proof of this result, see \citet[][Theorem~2.2.30]{DeZe10a} or \citet[][Fact~1]{HMST21a} in the present context for the case of two states.
We provide more details in \Cref{sec:app-single-agent}.

Now consider any number of agents with public signals: each agent observes all other agents' signals.
If the agents' signals are conditionally independent and identically distributed across periods and agents, then the signals of $n$ agents in a single period are as informative as those of a single agent over $n$ periods.  
Hence, $n$ agents observing each other's signals and choosing actions optimally learn $n$ times as fast as a single agent in autarky, and the rate of learning is $n\raut$, where $\raut = \raut^i$ does not depend on $i$.
\begin{align*}
	\pr[a_t^i \neq a_\omega] = e^{-n \raut t + o(t)}
\end{align*}
In particular, the rate of learning grows linearly in the number of agents and can thus become arbitrarily large provided there are sufficiently many agents.
By contrast, when the signals are not conditionally independent across agents, even two agents can learn much faster than a single agent with the same distribution of private signals. 
For our leading example (cf.\ \Cref{sec:leading-example}), if the probability that both agents simultaneously receive the wrong signal is sufficiently low, then an observer of the joint signals can achieve an arbitrarily high learning rate.

\section{Coordinated Learning}\label{sec:coordinated-learning}

We study how fast agents can learn from observing their private signals and their neighbors' actions.
First, we take the perspective of a social planner who can design the agents' strategies and aims to maximize the learning rate of the slowest-learning agent.
In particular, we ask whether one can design strategies such that, in a sufficiently large and strongly-connected network, each agent's learning rate exceeds a given fixed bound.

Our first result shows that there is an upper bound on the learning rate of the slowest-learning agent that only depends on the marginal distributions of the agents' private signals.
In particular, it is independent of the number of agents, the network structure, and the strategy profile imposed by the social planner.
Hence, information aggregation breaks down under very general conditions, even in the favorable case when strategies are not constrained by incentives.
The driving force is a fundamental tradeoff between choosing the correct action and using actions to communicate information.
More precisely, each agent's strategy needs to trade off between choosing the action that is most likely to be correct in the current period and using actions to inform other agents about her private signals and her neighbors' actions to reduce others' probability of mistakes in future periods. 
The fact that learning is bounded shows that there is no way to achieve both objectives to a high degree simultaneously.

\begin{restatable}[Learning is bounded]{theorem}{main}\label{thm:main}
	For any number of agents $n$, any network, and any strategies $\sigma^1,\dots,\sigma^n$, some agent learns at a rate of at most $\rbdd = \min_{f\neq g} \max_{i\in N} \evcond[f]{\ell_{f,g}^i}$.
\end{restatable}

In other words, for any strategy profile, there is some agent $i$ such that for each $\epsilon > 0$, we have $\pr[a_t^i\neq a_\omega] \geq e^{-(\rbdd + \epsilon)t}$ for infinitely many periods $t$.
So fast learning by some agents always comes at the cost of other agents' learning.
It is clear from the expression for $\rbdd$ that it only depends on the marginal signal distributions $\mu_f^i$ and is determined by the two states that are hardest to distinguish.
\Cref{sec:proof-outlines} provides an outline of the proof of \Cref{thm:main}. 

To provide some intuition, we discuss two policies that a social planner may try to use to achieve a uniformly high learning rate.
In the first, each agent randomizes each period: with probability $10\%$, the agent's action choice reveals all of her information to her neighbors; with probability $90\%$, she chooses myopically optimal given her information.\footnote{Recall that we allow the set of actions to be infinite, so that a single action choice can encode all information available to the agent.} 
	These strategies provide each agent with all information (including everyone's private signals) about all but the most recent periods with high probability.
	Hence, the mistake probability is close to the public signal case in $90\%$ of periods in the limit. 
	However, each agent is likely to make a mistake in the remaining $10\%$ of periods, so that the probability of mistakes does not converge to $0$ and the learning rate is $0$.
	The second planner policy tries to improve on the first through less frequent information revelation.
	Now the probability that an agent's action in period $t$ reveals all of her information is, say, $e^{-10t}$; otherwise the strategies are as above.
	Thus, revelation periods are exponentially unlikely and do not prohibit learning at a rate above $\rbdd$ (assuming $10 > \rbdd$).
	However, the total number of revelation periods is almost surely finite under these strategies by the Borel-Cantelli lemma.
	So information revelation always ceases eventually, and the agents eventually behave purely myopically. 
	In the first example, enough information is revealed to allow for fast learning, but the revelation induces too many mistakes; the opposite holds for the second example.
	The proof of \Cref{thm:main} shows that no planner policy can provide enough information revelation to facilitate fast learning while avoiding too many mistakes in the process.

We remark in brief that \Cref{thm:main} can be upgraded to a statement about the fastest-learning agent's rate in symmetric situations.
Indeed, for any symmetric distribution of signal profiles, any symmetric network, and any symmetric strategy profile, all agents learn at the same rate.\footnote{An automorphism of a network is a permutation of agents that preserves neighborhood relations. Then, a network is symmetric if for any two agents $i$ and $j$, there exists an automorphism mapping $i$ to $j$.}
The canonical example for this situation is the complete network with signals that are conditionally i.i.d.\ across periods and agents and the same strategy for all agents.
It then follows directly from \Cref{thm:main} that no agent can exceed the rate $\rbdd$.
This gives one answer to a conjecture of \citet{HMST21a}, who speculate that a social planner can improve learning by instructing agents to overweight their private signals.
The intuition is that the actions become more informative about private signals and so provide more information for future periods.\footnote{The conjecture is also supported by results of \citet{ABM+25a}, who show that, in the classical herding model where agents arrive sequentially, moderate condescension---a misspecification where agents overweight their private signal but are otherwise rational---can significantly improve the efficiency of learning.}
The preceding remarks show that, assuming symmetry, the benefit of such a planner policy cannot exceed $\rbdd$.
In \Cref{sec:equilibrium-learning} we discuss another case in which \Cref{thm:main} can be upgraded to a statement about the fastest-learning agent's rate.

\Cref{thm:main} is tight in two respects. 
First, without additional assumptions on the strategies, the statement cannot be lifted from one slow-learning agent to all agents learning slowly. 
Indeed, it is easy to construct strategies for which a large fraction of agents learn very fast if there are many agents.
Consider the leading example from \Cref{sec:leading-example} with the following strategies: in each period, each odd-numbered agent chooses the action that matches her private signal in that period, and each even-numbered agent chooses optimally based on the actions of the odd-numbered agents. Then, the odd-numbered agents do not learn at all, and each even-numbered agent learns at a rate of $\frac n2\raut$, which grows linearly with the number of agents.

Second, the upper bound on the learning rate $\rbdd$ of the slowest-learning agent cannot be improved. 
This is the content of our second result.
In its proof, we construct strategies for which all agents learn at a rate close to $\rbdd$ or faster, provided the agents' signals are conditionally independent and identically distributed, the network is strongly connected, and the number of agents is large.
Since $\rbdd > \raut$, this shows that coordination can improve the learning rate of all agents.

\begin{restatable}[Coordination improves learning]{theorem}{learninglowerbound}\label{thm:coordination}
	Assume the signals are conditionally independent and identically distributed across periods and agents.
	For any $\epsilon > 0$, there is $n_0$ such that for all $n \ge n_0$ and any strongly connected network, there exist strategies $\sigma^1,\dots,\sigma^n$ such that each agent learns at a rate of at least $\rbdd - \epsilon$.
\end{restatable}

In other words, for any positive $\epsilon$, there exist strategies so that $\pr(a_t^i\neq a_\omega) \le e^{-(\rbdd -\epsilon)t + o(t)}$ for each agent $i$ if the number of agents is large enough.
The assumptions that the signals are conditionally independent and the network is strongly connected are both needed.\footnote{First, if the agents' signals were perfectly correlated instead of conditionally independent, no agent could exceed the autarky learning rate since even observing other agents' signals would give no information beyond the agent's private signals. Second, in networks that are not strongly connected, some agents may have no neighbors, in which case they have no more information than a single agent in autarky.}
We sketch the proof of \Cref{thm:coordination} in \Cref{sec:proof-outlines}.

We illustrate the bounds from \Cref{thm:main} and \Cref{thm:coordination} for the example in \Cref{sec:leading-example}.
Recall that there are two states $f$ and $g$ and two signals, and each agent's signal in each period matches the state with probability $p\in(\frac12,1)$.
A calculation shows that 
\begin{align*}
	\evcond[f]{\ell^i_{f,g}} = \evcond[g]{\ell^i_{g,f}} = (2p-1) \log\frac{p}{1-p}
\end{align*}
for each agent $i$.
Hence, $\rbdd = (2p-1) \log\frac{p}{1-p}$.
The expression for $\raut$ involves a minimization problem and cannot be stated in closed form.
For $p = 0.75$, we have numerically that $\raut \approx 0.144$ and $\rbdd \approx 0.549$. 
Hence, for these signal distributions, observing conditionally independent signals of four agents allows for faster learning than the slowest-learning agent in any network for any number of agents with any strategies.
\Cref{fig:learning-rates} illustrates the bounds on the learning rates for other values of $p$.

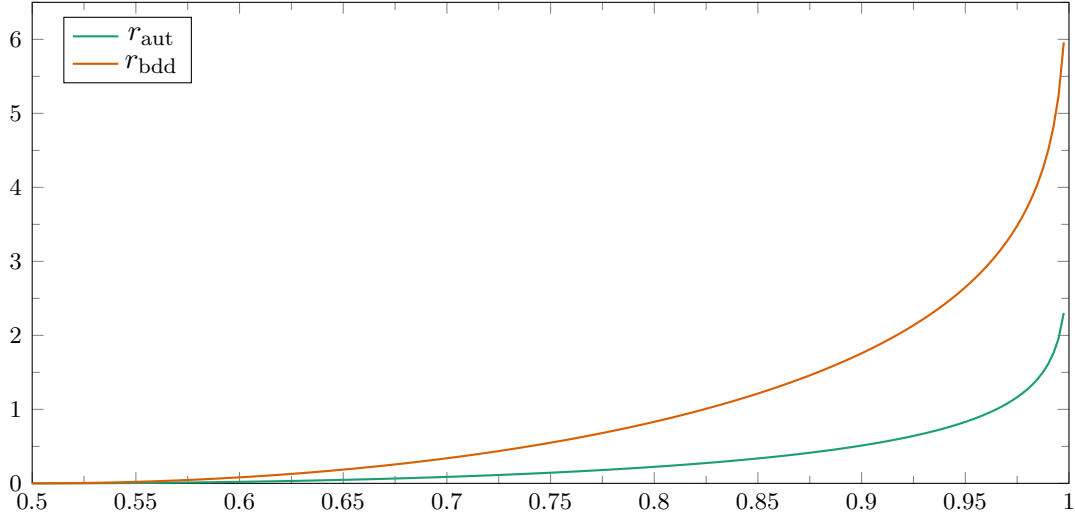
\begin{figure}[tb]
	\centering
	\pgfplotstableread[col sep = comma]{2states-2signals-symmetric-200points.dat}{\table}

	\tikzsetnextfilename{learningrates}
	\begin{tikzpicture}
	\tikzset{
		pin distance=8pt,
       every pin/.style={rectangle,rounded corners=3pt,font=\tiny},
   }
	\begin{axis}[
	   clip=false,
	   cycle list/Dark2-3, 
	   cycle multi list={
	              {solid,thick,line join=round},{dashed,line join=round}\nextlist
	              auto 
	          }, 
		xmin = 0.5, xmax = 1,
		ymin = 0, ymax = 6.5,
		xtick distance = .05,
		ytick distance = 1,
		minor tick num = 1,
		width = \textwidth-20pt,
		height = 0.5*\textwidth, 
		legend pos = north west,
		label style={font=\footnotesize},
		tick label style={font=\footnotesize},
		xlabel={},
		ylabel={},
	]

	\addplot +[y filter/.expression={y>10 ? NaN : y}] table [x = {q}, y = {raut}] {\table};
	\addlegendentry{$\raut$}
	\addplot +[y filter/.expression={y>10 ? NaN : y}] table [x = {q}, y = {rmaj}] {\table};
	\addlegendentry{$\rbdd$}

	\end{axis}
	\end{tikzpicture}
	\caption{Bounds on the learning rates for the example in \Cref{sec:leading-example}. The horizontal axis corresponds to the parameter $p$ determining the informativeness of the signal distributions. Higher values of $p$ correspond to more informative signals. The vertical axis indicates the rate of learning: $\raut$ is the learning rate of a single agent learning in autarky obtained from \Cref{prop:autarky} in \Cref{sec:app-single-agent}; $\rbdd$ is the upper bound on the learning rate of the slowest-learning agent in an arbitrary network and with arbitrary strategies obtained from \Cref{thm:main}; and for any strongly connected network with a large number of agents with the strategies described in the proof of \Cref{thm:coordination} (and sketched in \Cref{sec:proof-outlines}), each agent learns at a rate close to $\rbdd$.}
	\label{fig:learning-rates}
\end{figure}

\section{Equilibrium Learning}\label{sec:equilibrium-learning}

We turn to learning in equilibrium with geometrically discounting agents.
More precisely, suppose that all agents share a common discount factor $\delta \in [0,1)$.\footnote{The assumption that the discount factor is the same for all agents is purely for notational simplicity. All results remain valid with heterogeneous discount factors.}
The expected utility of agent $i$ for a strategy profile $\sigma$ is 
\begin{align*}
	u^i(\sigma) = \sum_{t\in T} \delta^{t-1} \ev[u(a_t^i,\omega)]
\end{align*}
where $a_t^i$ is $i$'s action in period~$t$ for the strategy profile $\sigma$.
The case $\delta = 0$ corresponds to myopic agents.
A strategy profile is a Nash equilibrium if no agent can increase her expected utility by unilaterally changing her strategy.
Since mixed strategy profiles map one-to-one to behaviorally equivalent pure strategy profiles for larger signal spaces (cf.\ \Cref{fn:pure-strategies}), it suffices to establish a bound on the learning rate for pure strategy equilibria.

\citet{HST24a} show that in any strongly connected network, all agents learn at the same rate in any equilibrium.
This implies that any upper bound on the slowest-learning agent's rate is also an upper bound on the fastest-learning agent's rate.\footnote{The proof of Lemma 2 of \citet{HST24a} does not make use of the fact that signals are conditionally independent across agents and thus applies in the current setting.}

\begin{lemma}[\citealp{HST24a}, Lemma 2]
	\label{lem:same-rate}
	For any number of agents, any strongly connected network, and any discount factor $\delta\in[0,1)$, all agents learn at the same rate in any equilibrium.
\end{lemma}

This lemma and the bound on the learning rate of the slowest-learning agent from \Cref{thm:main} show that no agent can learn at a rate faster than $\rbdd$ in any equilibrium for any number of agents in any strongly connected network.

\begin{corollary}\label{cor:equilibrium1}
	For any number of agents, any strongly connected network, and any discount factor $\delta\in[0,1)$, each agent learns at a rate of at most $\rbdd$ in any equilibrium.
\end{corollary}
This extends Theorem 1 of \citet{HST24a} to the case when the signals are not conditionally independent across agents and achieves a tighter bound on the learning rate.\footnote{The bound on the learning rate obtained by \citet{HST24a} is $2\max_{f\neq g}\max_{i\in N}\sup_{s\in S} |\ell^i_{f,g}(s)| \ge \rbddt > 2\rbdd$. 
However, their model differs slightly from ours. 
They allow the signal space and the distribution of signals to vary across periods and agents as long as the log-likelihood ratios of signals are bounded uniformly over states, agents, and periods. On the other hand, they assume that the action space and the signal space are finite and that signals are conditionally independent across periods and agents.}
\citet{HMST21a} obtain a smaller upper bound than \Cref{cor:equilibrium1} for two states and complete networks when the signals are conditionally independent and identically distributed across periods and agents, and the agents are myopic ($\delta = 0$).
The main feature of these results is that each agent's learning rate is bounded from above independently of the number of agents, showing that all but a vanishing fraction of private information is lost in large networks.
Giving a nontrivial lower bound on the equilibrium learning rate analogous to \Cref{thm:coordination} is an open problem (cf.\ \Cref{rem:lower-bound-equilibrium}).

For the rest of this section, we depart from the model and assume that agents observe their neighbors' actions and signals, and ask if this improves equilibrium learning.
In this alternative model, any nontrivial bound on the learning rate has to depend on the network structure. 
Indeed, if a rational agent observes everyone's signals, she learns at a rate of $n\raut^i$, which is not bounded independently of $n$.
The problem becomes interesting when the neighborhoods have bounded size.
Denote by $\Delta = \max_{i\in N} |N^i|$ the size of the largest neighborhood.
We show that even if agents observe their neighbors' actions and signals, the learning rate of each agent is at most $\Delta \rbdd$ in any equilibrium and any strongly connected network for any number of agents.
This result assumes that signals are conditionally independent and identically distributed across periods and agents. 
\begin{corollary}\label{cor:locally-public-signals}
	Assume the signals are conditionally independent across periods and agents.
	For any number of agents, any strongly connected network, and any discount factor $\delta\in[0,1)$, each agent learns at a rate of at most $\Delta \rbdd$ in any equilibrium.
\end{corollary}
First, it is clear from \Cref{lem:same-rate} that all agents learn at the same rate in any equilibrium.
It thus suffices to show that some agent learns at most at the claimed rate. 
To this end, we embed the setup of \Cref{cor:locally-public-signals} in the model where each agent only observes her neighbors' actions but not their signals.
For each $i\in N$ and $t\in T$, let $\tilde {\mathfrak s}_t^i = (\mathfrak s_t^j)_{j\in N^i}$ be the vector with the signals of all of $i$'s neighbors in period $t$.
The distribution of $\tilde{\mathfrak s}_t^i$ in state $f\in\Omega$ is the $|N^i|$-fold product $\tilde\mu_f^i = (\mu_f^i)^{\otimes |N^i|}$ of $\mu_f^i$.
Thus, defining $\tilde \ell_{f,g}^i = \log \frac{d\tilde\mu_f^i}{d\tilde\mu_{g}^i}(\tilde{\mathfrak s}_1^i)$ for two states $f,g\in\Omega$, we have $\tilde \ell_{f,g}^i = \sum_{j\in N^i} \ell_{f,g}^j$ since the signals $(\mathfrak s_1^j)_{j\in N^i}$ are conditionally independent.
Thus,
\begin{align*}
	\min_{f\neq g} \max_{i\in N} \evcond[f]{\tilde\ell_{f,g}^i} = \min_{f\neq g} \max_{i\in N} \sum_{j\in N^i} \evcond[f]{\ell_{f,g}^j} \le \max_{i\in N} |N^i| \min_{f\neq g} \max_{i\in N} \evcond[f]{\ell_{f,g}^i} = \Delta \rbdd
\end{align*}
Applying \Cref{thm:main} to the modified signals $(\tilde{\mathfrak s}_t^i)_{i\in N,t\in T}$ gives that some agent learns at a rate of at most $\Delta\rbdd$, concluding the proof.
Note that the modified signals are not conditionally independent across agents so that the full generality of \Cref{thm:main} was needed.

\section{Proof Outlines}\label{sec:proof-outlines}

In this section, we sketch the proofs of the two main results.

\subsection{About the Proof of \Cref{thm:main}}
\label{sec:proof-thm1}

In contrast to earlier results, \Cref{thm:main} provides an upper bound on the learning rate for arbitrary strategy profiles rather than only equilibria and is thus, in principle, harder to establish. 
However, the greater generality makes clear that the argument cannot rely on analyzing belief dynamics or incentives.
Moreover, it is clearly without loss to assume that the network is complete, which is not obvious for equilibrium learning.

The proof of \Cref{thm:main} greatly simplifies when additionally assuming that the signals are conditionally independent across agents, requiring that the signals' log-likelihood ratios are bounded, and loosening the bound on the learning rate.
This case already captures the main tension in the model.
We sketch its proof and then elaborate on how to adapt it to get \Cref{thm:main}. 

\begin{restatable}[Learning is bounded, weak form]{manualtheorem}{mainindependent}\label{thm:independent-signals}
	Assume the signals are conditionally independent across periods and agents, and the log-likelihood ratio $\ell_{f,g}^i(s)$ is bounded for each agent $i$ and any two states $f,g$. For any number of agents $n$, any network, and any strategies $\sigma^1,\dots,\sigma^n$, some agent learns at a rate of at most $\rbddt = 2\min_{f\neq g} \max_{i\in N} \sup_{s\in S} |\ell_{f,g}^i(s)|$.
\end{restatable}

We focus on the case of two states $\{f,g\}$.
The intuition is as follows.
Suppose all agents learn faster than the prescribed rate.
Then in state $g$, each agent chooses $a_g$ except for a set of signal trajectories with very small probability in any sufficiently late period.
But then in state $f$, these sets still have a moderately small probability, and so with positive probability, all agents incorrectly choose $a_g$ in each sufficiently late period.
This contradicts that agents learn at all.

	In more detail, assume for contradiction that all agents learn at a rate of at least $\rbddt + 3\epsilon$ for some positive $\epsilon$.
A preliminary lemma turns the limit defining the learning rate into a statement about each sufficiently late period at a small cost in the learning rate.
More precisely, \Cref{lem:Ht0} exhibits a history $H_{\le t_0}$ with positive probability in state $g$ (and thus also in state $f$) such that $$\prcond[g]{a_t^i\neq a_g \mid H_{\le t_0}} \le e^{-(\rbddt + 2\epsilon)t}$$ for each agent $i$ and each period $t > t_0$, and $t_0$ may be chosen arbitrarily large.
In the proof, we show that there exists a set of signal profile trajectories $\mathcal S_{\le t_0}\subset S^{N\times t_0}$ such that (i) $\prcond[g]{\mathcal S_{\le t_0}} > 0$, (ii) each $s_{\le t_0}\in\mathcal S_{\le t_0}$ is consistent with $H_{\le t_0}$, and (iii) for each $i\in N$, each $t\in T$, and almost every $s_{\le t_0}\in\mathcal S_{\le t_0}$,
\begin{align*}
	\prcond[g]{a_t^i\neq a_g \mid s_{\le t_0}} \le e^{-(\rbddt + \epsilon)t}
\end{align*}

Denote by $H$ the infinite history that follows $H_{\le t_0}$ up to $t_0$ and for which each agent chooses $a_g$ in each period after $t_0$.
We have that $\prcond[g]{H \mid s_{\le t_0}} \ge \frac12$ for almost every $s_{\le t_0}\in\mathcal S_{\le t_0}$ if $t_0$ is large enough, since the probability that some agent chooses an action different from $a_g$ in any period after $t_0$ is small by the preceding inequality. 
We say that agent $i$ defects in period $t> t_0$ if $i$ is the first agent to deviate from $H$, i.e., if play follows $H_{<t}$ in the first $t-1$ periods and $a_t^i \neq a_g$.
Consider the set $D_t^i$ of those trajectories of $i$'s signals for which $i$ defects in period $t$.
We aim to bound the probability of each $D_t^i$ conditional on state $f$ and signal profile trajectories in $\mathcal S_{\le t_0}$.

If $i$ receives signals in $D_t^i$ and no other agent defects before $t$, then $a_t^i \neq a_g$, and so 
\begin{align*}
	\prcond[g]{D_t^i, \forall j \neq i, H^j_{<t} \mid s_{\le t_0}} \le \prcond[g]{a_t^i \neq a_g \mid s_{\le t_0}} \le e^{-(\rbddt + \epsilon)t}
\end{align*}
for almost every $s_{\le t_0}\in \mathcal S_{\le t_0}$.
Since signals are conditionally independent across periods and agents, the probability on the left-hand side factors.
\begin{align*}
	\prcond[g]{D_t^i, \forall j \neq i, H^j_{<t} \mid s_{\le t_0}} = \prcond[g]{D_t^i \mid s_{\le t_0}} \prcond[g]{\forall j \neq i, H^j_{<t} \mid s_{\le t_0}}
\end{align*}
for almost every $s_{\le t_0}\in\mathcal S_{\le t_0}$.
But
\begin{align*}
	\prcond[g]{\forall j \neq i, H^j_{<t} \mid s_{\le t_0}} \geq \prcond[g]{H \mid s_{\le t_0}} \ge \frac12
\end{align*}
and so $\prcond[g]{D_t^i \mid s_{\le t_0}} \le 2e^{-(\rbddt + \epsilon)t}$ for almost every $s_{\le t_0}\in\mathcal S_{\le t_0}$.
Since signals are conditionally independent across agents, we have $\prcond[g]{D_t^i \mid s_{\le t_0}} = \prcond[g]{D_t^i \mid s_{\le t_0}^i}$ and similarly for $f$ instead of $g$.
Hence, from the definition of $\rbddt$ and the conditional independence of signals across periods, we see that 
\begin{align*}
	\prcond[f]{D_t^i \mid s_{\le t_0}} = \prcond[f]{D_t^i \mid s_{\le t_0}^i } \le e^{\rbddt t} \prcond[g]{D_t^i \mid s_{\le t_0}^i } = e^{\rbddt t} \prcond[g]{D_t^i \mid s_{\le t_0}} \le 2 e^{-\epsilon t}
\end{align*}
for almost every $s_{\le t_0}\in\mathcal S_{\le t_0}$, since the probability of $D_t^i$ increases by a factor of at most $e^{\frac12 \rbddt t}$ and the probability of $s_{\le t_0}^i$ decreases by a factor of at most $e^{\frac12 \rbddt t_0}$ when conditioning on $f$ rather than $g$.
If $t_0$ is large enough, then $2n \sum_{t > t_0} e^{-\epsilon t} \le \frac{1}{2}$, and so the probability that some agent defects after a signal profile trajectory in $\mathcal S_{\le t_0}$ is thus at most $\frac12$ even in state $f$.
But then, with positive probability, all agents choose the incorrect action $a_g$ in all periods after $t_0$ in state $f$.
Thus, in state $f$, each agent's mistake probability in each period after $t_0$ is bounded below by a positive constant, so her learning rate is $0$, contradicting the assumption.

The proof of \Cref{thm:main} needs to address two further points. 
First, it needs an argument that only loses $e^{\rbdd t}$ instead of $e^{\rbddt t}$ in the last inequality above.
This is taken care of with a simple trick.
Second, one cannot use the conditional independence of the agents' private signals.
This is a substantial hurdle that requires a delicate argument.

\subsection{About the proof of \Cref{thm:coordination}}
\label{sec:proof-thm2}

To prove \Cref{thm:coordination}, we construct strategies for which each agent learns at a rate close to $\rbdd$ or faster.
For complete networks, such strategies are easy to describe: each agent follows her private signals as long as those are sufficiently decisive, and otherwise she follows the action taken by most agents in the previous period.

In more detail, consider again the case of two states $f$ and $g$.
Agent $i$'s log-likelihood ratio for $f$ over $g$ in period~$t$ given the agent's private signals $\mathfrak s_{\le t}^i$ is close to $\evcond[f]{\ell_{f,g}^i} t$ with high probability in state $f$, and similarly in state $g$.
If this log-likelihood ratio for $f$ over $g$ is larger than $(\evcond[f]{\ell_{f,g}^i} - \epsilon) t$, agent $i$ chooses $a_f$ independently of the other agents' actions; similarly, if $i$'s log-likelihood ratio for $g$ over $f$ is larger than $(\evcond[g]{\ell_{g,f}^i} - \epsilon) t$, she chooses $a_g$; otherwise, she chooses the action that most agents chose in period $t-1$.

These strategies lead to faster learning than autarky.
First, agents decide based on their private signals only if those clearly favor one state, and so in that case, mistakes are less likely than when always relying on one's private signals. 
Second, most agents decide based on their private signals for late periods since then each agent's likelihood ratio is close to its expectation with high probability, and each of these actions is correct with high probability independently of the others. 
Hence, the most popular action in any period is very likely to be correct if there are many agents.
So either case improves upon learning in autarky.
Note, however, that each agent could unilaterally achieve a higher rate than $\rbdd$ by following the most popular action in each period.
In particular, these strategies are not an equilibrium for rational and geometrically discounting agents.

For a strongly connected but not necessarily complete network, one can use a similar idea.
The only obstacle is that agents do not know the most popular action of the previous period if they cannot see all other agents' actions.
We deal with this issue by dedicating some periods to propagating information about the agents' actions in previous periods through the network.
More precisely, in these propagation periods, each agent mimics the action of a neighbor in a recent previous period, or of a neighbor's neighbor (which they know through their neighbors' mimicking behavior), and so on for any distance level.
In the remaining periods, agents act as before: they follow their private signals if those are sufficiently decisive, and otherwise they follow the action taken by most agents in the latest non-propagation period, which they have learned in the intervening propagation periods.

\section{Discussion}\label{sec:discussion}

We conclude with several remarks about assumptions, model variations, and open problems.

\begin{remark}[Genericity of the utility function]\label{rem:non-unique-optimal-actions}
	We have assumed that the agents' utility functions are suitably generic, i.e., that no action is optimal in two different states and there is a unique optimal action in each state. 
	The first assumption is necessary to make the problem interesting: if the same action is optimal in all states, there is no need for information, and all agents can choose optimally from the first period onward.
	The second assumption forces a tradeoff between correct action choices and using actions as messages.
	By contrast, if there are two optimal actions in each state and signals are binary, then each agent can choose an action optimally based on her available information and simultaneously communicate her private signal in each period.
	Thus, actions are sufficient for identifying private signals and learning is as fast as if agents observed their neighbors' private signals, so that \Cref{thm:main} fails.
	Note that the above strategies are even in equilibrium if the network is complete.\footnote{This equilibrium is reminiscent of a construction by \citet{HRS15a}, who consider bandit problems where agents observe each other's actions, but not the payoffs.
	They allow agents to communicate via cheap talk messages and show that cheap talk equilibria can replicate any equilibrium with publicly observable payoffs. 
	In our model, a multiplicity of optimal actions enables similar cheap talk communication and can restore the public information case in equilibrium.}
\end{remark}

\begin{remark}[Arbitrary signal correlation]\label{rem:correlated-signals}
	We have assumed that the conditional distributions of signal profiles are mutually absolutely continuous and that the signals are conditionally independent across periods.
	\Cref{thm:main} breaks down emphatically without these assumptions. 
	
	Appropriate correlation across agents or periods allows even a small number of agents or a single agent in a small number of periods to learn the state with certainty.
	For example, consider the instance in \Cref{sec:leading-example} with signal precision $p = \frac23$.
	If there are three agents and the distribution of signal profiles is such that exactly two agents receive a signal matching the state conditional on either state, and each agent chooses the action matching her signal in the first period, then all agents know the state from the second period onward and can choose optimally in all future periods.
	For correlation across periods, consider a single agent who receives exactly two signals matching the state in the first three periods.
	This reveals the state in the first three periods. 
\end{remark}

\begin{remark}[Lower bounds for equilibrium learning]\label{rem:lower-bound-equilibrium}
	\Cref{thm:coordination} shows that non-trivial information aggregation is possible if the agents follow prescribed strategies. 
	However, it remains an open problem whether there is any equilibrium for any network in which every agent learns faster than a single agent in autarky.
	Answering this question would likely require either new conceptual insights into equilibrium learning or explicitly constructing an equilibrium in which learning exceeds the autarky benchmark.
	
	A related question is whether the mediation of the information exchanged between the agents can improve equilibrium learning.
	The mediator observes all agents' actions but not their private signals and sends a private message to each agent.
	A special case is designing the network structure.
	For example, it could be beneficial for equilibrium learning if the mediator rewards an agent for deviating from a consensus action by giving her access to more agents' actions in future periods.
	This incentivizes agents to act based on their private signals rather than herd, potentially improving equilibrium learning.
\end{remark}

\begin{remark}[Learning with experimentation]\label{rem:experimentation}
	We shut down experimentation motives by assuming that agents do not observe their flow utilities. 
	In \Cref{fn:experimentation}, we explain that this allows for noisy observations of utilities as long as those do not provide information beyond next period's private signal. 
	\Cref{thm:main} survives (with a different constant replacing $\rbdd$) even if agents make noisy observations of their utilities for which future signals are not a sufficient statistic. 
	Suppose whenever agent $i$ chooses action $a$ in state $f$, her observed utility, denoted $u_t^i(a)$, follows a distribution $\nu_{a,f}^i$ on $\mathbb R$.
	While this model allows for experimentation, it embeds into ours by replacing each private signal $\mathfrak s_t^i$ with $\tilde{\mathfrak s}_t^i = (\mathfrak s_t^i,(u_t^i(a))_{a\in A})$. 
	Assuming the augmented signals $(\tilde{\mathfrak s}_t^i)_{i\in N,t\in T}$ satisfy our hypothesis, \Cref{thm:main} still yields an upper bound on the learning rate.
	This is also an upper bound for the original instance with observed utilities $u_t^i(a)$ since the embedding provides each agent with weakly more information.
\end{remark}

\begin{remark}[Random networks]\label{rem:random-networks}
	In an extension of our model, the network is drawn randomly (and independently from the state and the signals) according to some distribution in each period.  
	The upper bound on the learning rate from \Cref{thm:main} clearly remains valid in this setting since it holds even for complete networks.
	\Cref{thm:coordination} also survives so long as all agents know each period's network and all networks are strongly connected.
	We sketch the necessary changes to the strategies in \Cref{fn:random-networks} in the proof of \Cref{thm:coordination}.
\end{remark}

\section*{Acknowledgements}
The author thanks seminar audiences in Berlin, Bonn, and Warwick, as well as Krishna Dasaratha, Johannes H\"orner, Kate Huang, and Ilan Kremer for comments on the paper.
The author also acknowledges support by the DFG under the Excellence Strategy EXC-2047.

\newpage
\appendix

\section*{APPENDIX}\label{sec:appendix}

\section{Preliminaries}\label{sec:app-prelims}

In this appendix, we recall a classic result from the theory of large deviations of random walks and start by setting up the notation.
 
The cumulant generating function of a random variable $\ell$ is
\begin{align*}
	\lambda(z) = \log \ev[e^{z\ell}]
\end{align*}
We assume that $\ell$ is non-degenerate and that $\lambda(z)$ is finite for each $z\in\mathbb R$, which implies that $\ev[\ell]$ is finite.
The Fenchel-Legendre transform of $\lambda$ is
\begin{align*}
	\lambda^*(\eta) = \sup_{z\in\mathbb R} \eta z - \lambda(z)
\end{align*}
We collect some properties of $\lambda$ and $\lambda^*$.
\begin{lemma}[\citealp{DeZe10a}, Lemma 2.2.5]\label{lem:lambda-star-properties} 
	Let $I^* = \{\eta\in\mathbb R\colon \exists z \in \mathbb R, \lambda'(z) = \eta\}$.
	Then,\footnote{Since we assume that $\ell$ is non-degenerate and has finite cumulants, \Cref{lem:lambda-star-properties} avoids some case distinctions compared to Lemma~2.2.5 of \citet{DeZe10a} and allows strengthening convexity of $\lambda$ to strict convexity. See also Lemma 5 of \citet{HMST21a}.}
	\begin{enumerate}[label=(\roman*)]
	\item $\lambda$ is strictly convex, and $\lambda^*$ is non-negative and convex.\label{item:lambda-star-properties-convex}
	\item For all $\eta \ge \ev[\ell]$,
	\begin{align*}
		\lambda^*(\eta) = \sup_{z \ge 0} \eta z - \lambda(z)
	\end{align*} 
	and for all $\eta \le \ev[\ell]$,
	\begin{align*}
		\lambda^*(\eta) = \sup_{z\le 0} \eta z - \lambda(z)
	\end{align*}
	In particular, $\lambda^*$ is non-decreasing on $[\ev[\ell],\infty)$ and strictly increasing on $I^*\cap [\ev[\ell],\infty)$, and it is non-increasing on $(-\infty,\ev[\ell]]$ and strictly decreasing on $I^* \cap (-\infty,\ev[\ell]]$.
	Moreover, $\lambda^*(\ev[\ell]) = 0$.\label{item:lambda-star-properties-monotonicity}
	\item $\lambda$ is differentiable with 
	\begin{align*}
		\lambda'(z) = \frac{\ev[\ell e^{z\ell}]}{\ev[e^{z\ell}]}
	\end{align*}
	and if $\lambda'(z) = \eta$, then $\lambda^*(\eta) = \eta z - \lambda(z)$.\label{item:lambda-differentiable}
\end{enumerate}
\end{lemma}
It follows from~\ref{item:lambda-differentiable} that $\lambda'(0) = \ev[\ell]$ and that $\lambda'(z) \rightarrow \sup \ell$ as $z\rightarrow\infty$ and $\lambda'(z)\rightarrow \inf \ell$ as $z\rightarrow-\infty$.
Hence, $I^* = (\inf \ell,\sup\ell)$, and so $I^*$ contains an open neighborhood of $\ev[\ell]$ (where we use that $\ell$ is non-degenerate).

Let $\ell_1,\ell_2,\dots$ be independent random variables with the same distribution as $\ell$, and for $t\in T$, let $L_t = \sum_{r \le t} \ell_r$.
The following result shows that the probability that $L_t$ is less than $\eta t$ (plus a lower-order term) decreases exponentially in $t$ at rate $\lambda^*(\eta)$ if $\eta$ is smaller than $\ev[\ell]$, and similarly if $\eta$ is larger than $\ev[\ell]$.
\begin{theorem}[\citealp{Cram38a}]\label{thm:large-deviations}
	Let $\eta\in\mathbb R$.
	If $\inf_{z\in\mathbb R} \lambda'(z) < \eta \le \ev[\ell]$, then
	\begin{align*}
		\pr[L_t \le \eta t + o(t)] = e^{-\lambda^*(\eta)t + o(t)}
	\end{align*}
	and if $\ev[\ell] \le \eta < \sup_{z\in\mathbb R} \lambda'(z)$, then
	\begin{align*}
		\pr[L_t \ge \eta t + o(t)] = e^{-\lambda^*(\eta)t + o(t)}
	\end{align*}
\end{theorem}
In the stated version, \Cref{thm:large-deviations} follows from Theorem~2.2.3 of \citet{DeZe10a} by recalling that $\lambda^*$ is non-increasing on $(-\infty,\ev[\ell]]$ and non-decreasing on $[\ev[\ell],\infty)$, or by applying Theorem~6 of \citet{HMST21a} to $\ell$ and $-\ell$.

For $f,g\in\Omega$ with $f\neq g$ and $i\in N$, let $\lambda^i_{f,g}$ be the cumulant generating function of $\ell^i_{f,g}$, and denote by $(\lambda^i_{f,g})^*$ its Fenchel-Legendre transform.
It is not hard to show that $\lambda^i_{f,g}(z) = \lambda^i_{g,f}(-(z+1))$ and $(\lambda^i_{f,g})^*(\eta) = (\lambda^i_{g,f})^*(-\eta) - \eta$.\footnote{These relations appear in Lemma~6 of \citet{HMST21a}. Since they use different sign conventions to define $\lambda^i_{f,g}$ and $(\lambda^i_{f,g})^*$, we reproduce the argument here. 
First,
\begin{align*}
	\lambda^i_{f,g}(z) = \log \int_S e^{z \log\frac{d\mu^i_f}{d\mu^i_{g}}(s)} d\mu^i_f(s) 
	= \log \int_S \left(\frac{d\mu^i_f}{d\mu^i_{g}}(s)\right)^z d\mu^i_f(s) 
	= \log \int_S \left(\frac{d\mu^i_{g}}{d\mu^i_f}(s)\right)^{-(z+1)}d\mu^i_{g}(s) = \lambda^i_{g,f}(-(z+1))
\end{align*}
Thus,
\begin{align*}
	(\lambda_{f,g}^i)^*(\eta) = \sup_{z\in\mathbb R} \eta z - \lambda^i_{f,g}(z) = \sup_{z\in\mathbb R} \eta z - \lambda^i_{g,f}(-(z+1)) = \sup_{z\in\mathbb R} (-\eta) z - \lambda^i_{g,f}(z) - \eta = (\lambda^i_{g,f})^*(-\eta) - \eta
\end{align*}
}
Hence, by \Cref{lem:lambda-star-properties}\ref{item:lambda-star-properties-monotonicity},
\begin{align}
	(\lambda^i_{f,g})^*(-\evcond[g]{\ell^i_{g,f}}) = (\lambda^i_{g,f})^*(\evcond[g]{\ell^i_{g,f}}) + \evcond[g]{\ell^i_{g,f}} = \evcond[g]{\ell^i_{g,f}}\label{eq:lambda-star-state-swap}
\end{align}
Since $\lambda^i_{f,g}(0) = 0$, $\lambda^i_{f,g}(-1) = \lambda^i_{g,f}(0) = 0$, and since $\lambda^i_{f,g}$ is strictly convex, $\lambda^i_{f,g}$ attains its minimum on $(-1,0)$.
Then, using again that $\lambda^i_{f,g}$ is strictly convex,
\begin{align}
	(\lambda^i_{f,g})^*(0) = \sup_{z\in\mathbb R} -\lambda^i_{f,g}(z) 
	= - \min_{z\in(0,1)} \lambda^i_{f,g}(z) 
	< (\lambda^i_{f,g})'(0) 
	= \evcond[f]{\ell^i_{f,g}}\label{eq:lambda-star-0-bound}
\end{align}

\section{Single-Agent Learning}\label{sec:app-single-agent}

We recall some results for a single agent learning in autarky and extend those to more than two states. 

When there are only two states, the rate of learning a single agent can achieve in autarky is well-known.  
It essentially follows from \Cref{thm:large-deviations} and appears as Fact 1 of \citet{HMST21a}.\footnote{Note that $0\in (I^i_{f,g})^*$ by the remarks after \Cref{lem:lambda-star-properties} and the fact that $\sup \ell^i_{f,g} > 0$ and $\inf \ell^i_{f,g} < 0$.}
\begin{proposition}[\citealp{HMST21a}, Fact 1]\label{prop:autarky}
	Let $\Omega = \{f,g\}$. 
	The rate of learning of agent $i$ in autarky is $(\lambda^i_{f,g})^*(0)$.
	More precisely, the probability that agent $i$ makes a mistake in period~$t$ when choosing actions optimally based on her private signals is
	\begin{align*}
		\pr[a_t^i\neq a_\omega] = e^{-(\lambda^i_{f,g})^*(0)t + o(t)}
	\end{align*}
\end{proposition}

For more than two states, the optimal rate of learning is determined by the two states that are the hardest to distinguish. 
More precisely, the optimal rate of learning equals the minimum of the optimal learning rates when restricting to any pair of states.
As for the case of two states, the ``maximum likelihood strategy'' achieves the highest learning rate: in any period, choose the action that is optimal in a most probable state.
This result can be obtained from Theorem~2.2.30 of \citet[][]{DeZe10a}.
We state it here along with a proof that is specific to our setting.
\begin{corollary}\label{cor:reduction-to-two-states}
	The probability that agent $i$ learning in autarky and choosing actions optimally makes a mistake in period~$t$ is
	\begin{align*}
		\pr[a_t^i\neq a_\omega] = e^{-\raut^i t + o(t)}
	\end{align*}
	where
	\begin{align*}
		\raut^i = \min_{f\neq g} (\lambda^i_{f,g})^*(0)
	\end{align*}
\end{corollary}
\begin{proof}
	Let 
	\begin{align*}
		r = \sup_{\sigma^i} \liminf_{t\rightarrow\infty} -\frac1t\log \pr[a_t^i \neq a_\omega]
	\end{align*}
	where $a_t^i = \sigma_t^i(\mathfrak s_1^i,\dots,\mathfrak s_t^i)$ and the supremum is taken over all strategies $\sigma^i = (\sigma_1^i,\sigma_2^i,\dots)$ of agent $i$ in autarky.
	Hence, $r$ is the optimal rate of learning, and we have to show that $r = \raut^i$.
	\setcounter{step}{0}
	\begin{step}[$r \le \raut^i$]
		Assume for contradiction that $r > \raut^i$, and let $\sigma^i$ be a strategy such that
		\begin{align*}
			\liminf_{t\rightarrow\infty} -\frac1t\log \pr[a_t^i \neq a_\omega] > \raut^i
		\end{align*}
		Then, there are $f,g\in\Omega$ with $f\neq g$ such that $\liminf_{t\rightarrow\infty} -\frac1t\log \pr[a_t^i \neq a_\omega] > (\lambda_{g,f}^i)^*(0)$.
		Note that $\sigma^i$ is also a strategy for the problem after restricting to the subset $\{f,g\}$ of states, and as such achieves a rate of learning of
		\begin{align*}
			\liminf_{t\rightarrow\infty} -\frac1t\log \pr[a_t^i \neq a_\omega \mid \omega \in\{f,g\}] \ge \liminf_{t\rightarrow\infty} -\frac1t\log \pr[a_t^i \neq a_\omega] > (\lambda_{g,f}^i)^*(0)
		\end{align*} 
		where the first inequality follows from the fact that $\pr[a_t^i \neq a_\omega \mid \omega \in\{f,g\}] \pr[\omega\in\{f,g\}] \le \pr[a_t^i\neq a_\omega]$ and $\pr[\omega\in\{f,g\}] > 0$.
		But this contradicts \Cref{prop:autarky}. 
	\end{step}
	\begin{step}[$r \ge \raut^i$]
		It suffices to find a strategy that achieves a learning rate of at least $\raut^i$.
		Let $\sigma^i = (\sigma_1^i,\sigma_2^i,\dots)$ be a strategy with 
		\begin{align*}
			a_t^i = \sigma_t^i(\mathfrak s_1^i,\dots,\mathfrak s_t^i) \in \{a_f\in A\colon f\in\Omega, L_{f,g,t}^{i} \ge 0 \;\forall g\in\Omega\}
		\end{align*}
		for all $t\in T$.
		Thus, $\sigma^i$ is a ``maximum-likelihood strategy'', i.e., it chooses the optimal action for a most probable state.\footnote{Note that $\sigma^i$ is well-defined since $L_{f,g,t}^{i} + L_{g,g',t}^{i} = L_{f,g',t}^{i}$ ensuring that the right-hand side in the definition of $\sigma^i$ is always non-empty.}
		Then, for all $f\in\Omega$,
		\begin{align*}
			\prcond[f]{a_t^i \neq a_f} \le \sum_{g\in\Omega} \prcond[f]{L_{f,g,t}^{i} \le 0} = \sum_{g\in\Omega} e^{-(\lambda_{f,g}^i)^*(0)t + o(t)} \le e^{-\raut^i + o(t)}
		\end{align*}
		where the first inequality follows from the definition of $\sigma^i$, the equality follows from \Cref{thm:large-deviations}, and the second inequality uses the definition of $\raut$ (and may require adjusting the lower-order term $o(t)$).
		Hence,
		\begin{align*}
			\pr[a_t^i\neq a_\omega] = \sum_{f\in\Omega} \prcond[f]{a_t^i\neq a_f} \pr[\omega = f] \le e^{-\raut^i + o(t)}
		\end{align*}
		as required.
	\end{step}
\end{proof}

\section{Omitted Proofs From \Cref{sec:coordinated-learning}}\label{sec:coordinated-learning-appendix}

In this section, we prove the upper and lower bound on the optimal learning rate claimed in \Cref{thm:main} and \Cref{thm:coordination}.

We start with three auxiliary lemmas.
Heuristically, the first states the following.
Fix a (one-dimensional) random walk with i.i.d.\ increments, an upper (lower) slope larger (smaller) than the expectation of each increment, and a probability threshold smaller than $1$.
Then, for each period $t_0$, the probability that the random walk remains inside an affine wedge with the given upper and lower slopes and sufficiently large intercepts in all periods after $t_0$ conditional on being well inside the wedge at $t_0$ exceeds the threshold.
\begin{lemma}\label{lem:wedge-probability}
	Let $\ell_1,\ell_2,\dots$ be non-degenerate i.i.d.\ random variables such that $\ev[e^{z\ell_1}]$ is finite for each $z\in\mathbb R$ and let $a^+,a^- \in \mathbb R$ with $a^+ > \ev[\ell_1]$ and $a^- < \ev[\ell_1]$.
	For $t \in T$, let $L_t = \sum_{r \le t} \ell_r$, and let $L = (L_1,L_2,\dots)$.
	Then, for each $\delta > 0$, there is $K > 0$ such that for each $b\in\mathbb R$, each $t_0 \in T$, and almost every (with respect to the distribution of $(L_1,\dots,L_{t_0})$) $x_{\le t_0} \in \mathbb R^{t_0}$ with $b-\frac{K}2 \le x_{t_0} \le b + \frac{K}{2}$,
	\begin{align*}
		\pr[L\in\mathcal W\mid L_{\le t_0} = x_{\le t_0}] \ge 1-\delta
	\end{align*}
	where $\mathcal W = \{x\in\mathbb R^T\colon \forall t\ge t_0, b-K + a^-(t-t_0) \le x_t \le b+K + a^+(t-t_0)\}$.
\end{lemma}
\begin{proof}
	Denote by $\lambda$ the cumulant generating function of $\ell_1$ and denote by $\lambda^*$ its Fenchel-Legendre transform.
	First, observe that the statement becomes stronger if $a^+$ and $a^-$ are closer to $\ev[\ell_1]$. 
	Hence, we may and will assume that $\ev[\ell_1] < a^+ < \sup_{z\in\mathbb R} \lambda'(z)$ and $\inf_{z\in\mathbb R} \lambda'(z) < a^- < \ev[\ell_1]$, so that \Cref{thm:large-deviations} applies.
	
	Fix $\delta > 0$.
	Let $b\in\mathbb R$, $t_0\in T$, and $x_{\le t_0}\in \mathbb R^{t_0}$ with $b-\frac{K}{2} \le x_{t_0} \le b+\frac{K}{2}$.
	Let $\tilde a^+ = \frac12(\ev[\ell_1] + a^+)$ and $\tilde a^- = \frac12(\ev[\ell_1] + a^-)$.
	By \Cref{lem:lambda-star-properties}\ref{item:lambda-star-properties-monotonicity}, $\lambda^*(\tilde a^+) > 0$ and $\lambda^*(\tilde a^-) > 0$.
	Thus, for all $t \ge t_0$, we have by \Cref{thm:large-deviations} that
	\begin{align*}
		&\pr[L_t - L_{t_0}\ge \tilde a^+(t-t_0) + o(t-t_0)] = e^{-\lambda^*(\tilde a^+)(t-t_0) + o(t-t_0)}\text{, and}\\
		&\pr[L_t - L_{t_0}\le \tilde a^-(t-t_0) + o(t-t_0)] = e^{-\lambda^*(\tilde a^-)(t-t_0) + o(t-t_0)}
	\end{align*}
	Thus, there is $t_1 \ge t_0$ such that
	\begin{align*}
		\pr[\forall t \ge t_1,  a^-(t-t_0) \le L_{t} - L_{t_0} \le a^+(t-t_0)] \ge 1-\frac{\delta}{2}
	\end{align*} 
	Let $K>0$ such that $\pr[-\frac{K}{2} \le \ell_1 \le \frac{K}{2}] \ge \frac{\delta}{2(t_1 - t_0)}$.
	Then it follows that 
	\begin{align*}
		\pr[\forall t \ge t_0, -\frac{K}{2} + a^-(t-t_0) \le L_{t} - L_{t_0} \le \frac{K}{2} + a^+(t-t_0)] \ge 1-\delta
	\end{align*} 
	But $\ell_1,\ell_2,\dots$ are independent and $b-\frac{K}{2} \le x_{t_0} \le b + \frac{K}{2}$ and so
	\begin{align*}
		\pr[L \in \mathcal W \mid L_{\le t_0} = x_{\le t_0}] 
		&\geq \pr[\forall t \ge t_0, - \frac{K}{2} + a^-(t-t_0) \le L_{t} - L_{t_0} \le \frac{K}{2} + a^+(t-t_0)] \ge 1-\delta
	\end{align*}
	which finishes the proof.
\end{proof}

The second lemma states the following.
Assume the agents follow strategies for which each of them learns at a strictly positive rate $r$. 
Then, there is a history $H_{\le t_0}$ up to $t_0$ such that conditional on $H_{\le t_0}$, each agent's probability of a mistake decreases exponentially at a rate close to $r$ from $t_0$ onward.
Heuristically, this turns the limit defining the learning rate into a statement about each sufficiently late period at a small cost in the learning rate. 
\begin{lemma}
	\label{lem:Ht0}
	Fix the number of agents $n$, a state $g\in\Omega$, and a learning rate $r > 0$.
	Let $\sigma^1,\dots,\sigma^n$ be strategies such that for each $i \in N$,
	\begin{align*}
		\liminf_{t\rightarrow\infty} -\frac1t \log \pr[a_t^i \neq a_g] \ge r
	\end{align*}
	Then, for each $\epsilon > 0$, there are $t_0 \in T$ and $H_{\le t_0}\in A^{N\times t_0}$ such that $\prcond[g]{H_{\le t_0}} > 0$ and for each $i \in N$ and $t > t_0$,
	\begin{align*}
		\prcond[g]{a_t^i \neq a_g \mid H_{\le t_0}} \le e^{-(r-\epsilon)t}
	\end{align*}
	Moreover, $t_0$ may be chosen arbitrarily large.
\end{lemma}
\begin{proof}
	By assumption, for each $i \in N$ and each $t \in T$,
	\begin{align*}
		\prcond[g]{ a_t^i \neq a_g } \le e^{-rt + o(t)}
	\end{align*}
	First, for $\tilde t_0\in T$ large enough,
	\begin{align*}
		\sum_{i\in N, t > \tilde t_0} \prcond[g]{a_t^i \neq a_g} \le \frac12
	\end{align*}
	and
	\begin{align*}
		\sum_{i\in N, t > \tilde t_0} \prcond[g]{a_t^i\neq a_g} = \sum_{H_{\le \tilde t_0} \in A^{N \times \tilde t_0}, \prcond[g]{H_{\le\tilde t_0}} > 0} \prcond[g]{H_{\le \tilde t_0}} \sum_{i\in N, t > \tilde t_0} \prcond[g]{a_t^i \neq a_g \mid H_{\le \tilde t_0}}
	\end{align*}
	and, thus, there is $H_{\le\tilde t_0}^*\in A^{N\times\tilde t_0}$ such that $\prcond[g]{H_{\le\tilde t_0}^*} > 0$ and
	\begin{align}
		\sum_{i\in N, t > \tilde t_0} \prcond[g]{a_t^i \neq a_g\mid H_{\le\tilde t_0}^*} \leq \frac12
		\label{eq:Ht01}	
	\end{align}
	Let $H\in A^{N\times T}$ such that $H_{\le\tilde t_0} = H_{\le\tilde t_0}^*$ and for each $i \in N$ and each $t > \tilde t_0$, $H_t^i = a_g$.
	Note that, by~\eqref{eq:Ht01},
	\begin{align}
		\prcond[g]{H \mid H_{\le \tilde t_0}} \ge \frac12	
		\label{eq:Ht02}
	\end{align}
	By the assumption of the lemma, for each $i\in N$, each $t_0 \geq \tilde t_0$ large enough, and each $t > t_0$,
	\begin{align*}
		\prcond[g]{a_t^i \neq a_g \mid H_{\le\tilde t_0}} \le \frac12 e^{-(r-\epsilon)t}
	\end{align*}
	and, thus, by~\eqref{eq:Ht02},
	\begin{align*}
		\prcond[g]{a_t^i\neq a_g \mid H_{\le t_0}} \le 2 \prcond[g]{a_t^i\neq a_g \mid H_{\le t_0}} \prcond[g]{H_{\le t_0} \mid H_{\tilde t_0}} \le 2\prcond[g]{ a_t^i \neq a_g \mid H_{\le\tilde t_0} } \le e^{-(r-\epsilon)t}
	\end{align*}
	which proves the claim.
\end{proof}

It is well-known that if $\mu_t,\nu_t$, $t\in\{1,2\}$, are probability measures such that $\nu_t$ is absolutely continuous with respect to $\mu_t$, then $\nu_1\otimes \nu_2$ is absolutely continuous with respect to $\mu_1\otimes \mu_2$, where $\nu_1\otimes \nu_2$ is the product measure on the corresponding product space.\footnote{A proof of this fact can be found here: \texttt{https://math.stackexchange.com/q/1042323}.}
Thus, since the distributions of signal profiles in different states are mutually absolutely continuous and signals are conditionally independent across periods, the same is true for the distributions of signal profile trajectories up to any period $t_0$.
This is the content of the third lemma.
We omit the easy proof.
\begin{lemma}
	\label{lem:absolute-continuity}
	For each $t \in T$ and each $f\in\Omega$, denote by $\mu_f^{\otimes t}$ the $t$-fold product measure of $\mu_f$ on $S^{N\times t}$.
	Then, for each $t\in T$ and each $f,g\in\Omega$, $\mu_{f}^{\otimes t}$ and $\mu_{g}^{\otimes t}$ are mutually absolutely continuous.
\end{lemma}

Now we are ready for the proof of \Cref{thm:main}.

\main*
\begin{proof}
	For each $g\in\Omega$, we denote by $\prcond[g]{\cdot\mid \mathfrak s = \cdot}$ a regular conditional probability, and write $\prcond[g]{\cdot \mid \mathfrak s = \cdot} = \prcond[g]{\cdot \mid \cdot}$ again for convenience.\footnote{That is, (i) for each $s\in S^T$, $\prcond[g]{\cdot\mid \mathfrak s = s}$ is a probability measure on the underlying probability space, (ii) for each event $E$, $\prcond[g]{E\mid \mathfrak s = \cdot}$ is measurable, and (iii) for each event $E$ and each measurable set $\mathcal S\subset S^{N\times T}$, $\prcond[g]{E \cap \mathfrak s^{-1}(\mathcal S)} = \int_{\mathcal S} \prcond[g]{E\mid \mathfrak s = s} d\pr{}\circ \mathfrak s^{-1}(s)$.}
	Let $f,g\in\Omega$ such that $\rbdd = \max_{i\in N} \evcond[f]{\ell_{f,g}^i}$.
	Fix strategies $\sigma^1,\dots,\sigma^n$.
	It is clearly without loss of generality to assume that $N^i = N$ for each $i\in N$.
	Assume for contradiction that there is $\epsilon > 0$ such that for each $i\in N$ and each $t\in T$,
	\begin{align}
		\pr[a_t^i\neq a_\omega] \le e^{-(\rbdd + 6\epsilon)t + o(t)}
		\label{eq:too-fast-learning}
	\end{align}
	We write $r = \rbdd$ in the rest of the proof for convenience.
	
	As promised by \Cref{lem:Ht0}, there are $t_0 \in T$ and $H \in A^{N\times T}$ such that 
	\begin{enumerate}[label=(\roman*)]
		\item $\prcond[g]{H_{\le t_0}} > 0$,\label{item:ht01}
		\item for each $i \in N$ and each $t > t_0$, $\prcond[g]{a_t^i \neq a_g \mid H_{\le t_0}} \le e^{-(r + 5\epsilon)t}$, \label{item:ht02}
		\item $\sum_{t > t_0} e^{-\epsilon t} < \frac{1}{32n^2}$ and $\epsilon t_0 \ge K$, and \label{item:ht03}
		\item for each $i \in N$ and each $t > t_0$, $H_t^i = a_g$, \label{item:ht04}
	\end{enumerate}
	where $K$ is the constant obtained from applying \Cref{lem:wedge-probability} to i.i.d.\ random variables $\ell_1,\ell_2,\dots$ each with distribution $\prcond[f]{\ell_{f,g}^i\in \cdot}$, $a^+ = \evcond[f]{\ell_{f,g}^i} + \epsilon$, $a^- = \evcond[f]{\ell_{f,g}^i} - \epsilon$, and $\delta = \frac{1}{2n}$ for each $i\in N$ separately and taking the maximum.\footnote{Here, $\prcond[f]{\ell_{f,g}^i\in \cdot}$ denotes the distribution of $\ell_{f,g}^i$ conditional on state $f$.}

	We proceed in multiple steps.
	\setcounter{step}{0}
	
	\begin{step}\label{step:correlated1}
		For each $i\in N$, each $t\in T$, and each $s_{\le t}^i \in S^t$, denote by $L_{g,f}^i(s_{\le t}^i)$ agent $i$'s log-likelihood ratio for $g$ over $f$ after observing the private signals $s_{\le t}^i$:
		\begin{align*}
			L_{g,f}^i(s_{\le t}^i) = \log\frac{\pi_0(g)}{\pi_0(f)} + \sum_{t' \le t} \ell_{g,f}^i(s_{t'}^i)
		\end{align*}
		
		First, we show that there is a set of signal profile trajectories $\mathcal S_{\le t_0} \subset S^{N\times t_0}$ up to $t_0$ and $m^1,\dots,m^n\in\mathbb R$ such that (i) $\mathcal S_{\le t_0} \subset \mathcal S(H_{\le t_0})$ (i.e., each trajectory of profiles in $\mathcal S_{\le t_0}$ induces $H_{\le t_0}$), (ii) $\prcond[g]{\mathcal S_{\le t_0}} > 0$, (iii) for each $i \in N$, each $t > t_0$, and almost every $s_{\le t_0}\in\mathcal S_{\le t_0}$, $\prcond[g]{a_t^i \neq a_g \mid  s_{\le t_0} } \le e^{-(r + 4\epsilon)t}$, and (iv) for each $i\in N$ and each $s_{\le t_0}\in\mathcal S_{\le t_0}$, $(m^i - \frac\epsilon2)t_0 \le L_{g,f}^i(s_{\le t_0}^i) \le (m^i+\frac\epsilon2)t_0$ (i.e., all trajectories in $\mathcal S_{\le t_0}$ induce roughly the same log-likelihood ratios).
		
		For each $i \in N$ and each $t > t_0$, let $\mathcal S_{\le t_0}(i,t) \subset \mathcal S(H_{\le t_0})$ be the set of signal profile trajectories up to $t_0$ that induce $H_{\le t_0}$ and after which the probability that agent $i$ does not choose $a_g$ in period $t$ and state $g$ is at least $e^{-(r + 4\epsilon)t}$.
		Formally, for each $i \in N$ and each $t > t_0$, let $\mathcal S_{\le t_0}(i,t)\subset \mathcal S(H_{\le t_0})$ such that for almost every $s_{\le t_0}\in\mathcal S_{\le t_0}(i,t)$,
		\begin{align*}
			\prcond[g]{a_t^i \neq a_g \mid  s_{\le t_0}} \ge e^{-(r + 4\epsilon)t}
		\end{align*}
		and for almost every $s_{\le t_0}\not\in\mathcal S_{\le t_0}(i,t)$, the reverse inequality holds.
		Then, for each $i \in N$ and each $t > t_0$,
		\begin{align*}
			e^{-(r + 4\epsilon)t} \prcond[g]{\mathcal S_{\le t_0}(i,t) \mid H_{\le t_0}}
			&\le \int_{\mathcal S_{\le t_0}(i,t)}\prcond[g]{a_t^i \neq a_g \mid  s_{\le t_0}} d\prcond[g]{ s_{\le t_0} \mid H_{\le t_0}} \\
			&\le \prcond[g]{a_t^i \neq a_g \mid H_{\le t_0}}\\
			&\le e^{-(r + 5\epsilon)t}
		\end{align*}
		and thus, $\prcond[g]{\mathcal S_{\le t_0}(i,t) \mid H_{\le t_0}} \le e^{-\epsilon t}$.
		Then, let 
		\begin{align*}
			\hat{\mathcal S}_{\le t_0} = \mathcal S(H_{\le t_0}) \setminus \bigcup_{i\in N, t > t_0} \mathcal S_{\le t_0}(i,t)
		\end{align*}
		for which we have
		\begin{align*}
			\prcond[g]{\hat{\mathcal S}_{\le t_0} \mid H_{\le t_0}} \ge 1 - \sum_{i\in N, t > t_0} \prcond[g]{\mathcal S_{\le t_0}(i,t) \mid H_{\le t_0}} \ge 1 - n\sum_{t > t_0} e^{-\epsilon t} > 0
		\end{align*}
		It thus follows from $\prcond[g]{H_{\le t_0}} > 0$ that $\prcond[g]{\hat{\mathcal S}_{\le t_0}} > 0$, and that for each $i\in N$, each $t > t_0$, and almost every $s_{\le t_0}\in\hat{\mathcal S}_{\le t_0}$,
		\begin{align*}
			\prcond[g]{a_t^i \neq a_g \mid  s_{\le t_0}} \le e^{-(r + 4\epsilon)t}
		\end{align*}
		by construction of $\hat{\mathcal S}_{\le t_0}$.
		Since $\mathbb R^n$ can be covered by countably many cubes with side length $\epsilon$, there are $\mathcal S_{\le t_0}\subset\hat{\mathcal S}_{\le t_0}$ and $m^1,\dots,m^n\in\mathbb R$ such that $\prcond[g]{\mathcal S_{\le t_0}} > 0$ and (iv) holds.
		Moreover, $\mathcal S_{\le t_0}$ satisfies (i) and (iii) since those are preserved under passing to subsets.
		This finishes the construction.
	\end{step}
	\begin{step}
		\label{step:correlated2}
		Second, let
		\begin{align*}
			\mathcal W =& \{ s\in S^{N\times T}\colon \forall i\in N,\forall t\ge t_0,\\
			&(m^i-\epsilon)t_0 - (\evcond[f]{\ell_{f,g}^i}+\epsilon)(t-t_0) \le L_{g,f}^i(s_{\le t}^i) \le (m^i + \epsilon)t_0 - (\evcond[f]{\ell_{f,g}^i} - \epsilon)(t-t_0)\} 
		\end{align*}
	 	be those trajectories of profiles such that the log-likelihood ratio for $g$ over $f$ of each agent $i$ based only on her private signals remains in a wedge with slopes $-\evcond[f]{\ell_{f,g}^i} + \epsilon$ and $-\evcond[f]{\ell_{f,g}^i} - \epsilon$ for every period from $t_0$ onward.
	 	By Property (iv) in \Cref{step:correlated1}, $\mathcal S_{\le t_0} \subset \mathcal W_{\le t_0}$.
	 	Moreover, by \Cref{lem:wedge-probability} and~\ref{item:ht03} in \Cref{step:correlated1}, for almost every $s_{\le t_0}\in \mathcal S_{\le t_0}$, $\prcond[f]{\mathcal W \mid  s_{\le t_0}} \ge 1 - \frac{n}{2n} = \frac12$.
		 	Note that for each $i\in N$, each $t\ge t_0$, and each $\mathcal V_{\le t}^i\subset \mathcal W_{\le t}^i$,
		 	\begin{align}
		 	\begin{aligned}
				e^{-(m^i+\epsilon)t_0 + (\evcond[f]{\ell_{f,g}^i} - \epsilon)(t-t_0)}\prcond[g]{\mathcal V_{\le t}^i} \le \prcond[f]{\mathcal V_{\le t}^i}
				\le e^{ -(m^i-\epsilon)t_0 + (\evcond[f]{\ell_{f,g}^i}+\epsilon)(t-t_0)}\prcond[g]{\mathcal V_{\le t}^i}
		 	\end{aligned}
		 	\label{eq:correlated-likelihood-increase}
		 	\end{align}
	\end{step}
	\begin{step}\label{step:correlated3}
		Third, we show that for each $i \in N$, each $t\ge t_0$, and almost every $s_{\le t_0}\in \mathcal S_{\le t_0}$,
		\begin{align}
			\prcond[g]{a_t^i \neq a_g \mid  s_{\le t_0}, H_{<t} } \le 2 e^{-(r + 4\epsilon)t}
			\label{eq:deviations-unlikely-Ht}
		\end{align}
		
		Fix $i \in N$ and $t > t_0$.
		By \Cref{step:correlated2}, Item~(iii) in \Cref{step:correlated1}, and the choice of $t_0$, for almost every $s_{\le t_0}\in\mathcal S_{\le t_0}$, 
		\begin{align*}
			\prcond[g]{H_{<t} \mid  s_{\le t_0}}
			&\geq 1 - \prcond[g]{\neg H_{<t} \mid  s_{\le t_0}}  \\
			&\geq 1 - \sum_{j\in N, t' > t_0} \prcond[g]{a_{t'}^j \neq a_g \mid  s_{\le t_0}} \\
			&\ge 1 - n\sum_{t' > t_0} e^{-(r + 4\epsilon)t'}\\
			&\ge \frac12
		\end{align*}
		and, moreover,
		\begin{align*}
			\prcond[g]{a_t^i \neq a_g \mid  s_{\le t_0}, H_{<t}} \prcond[g]{H_{<t} \mid  s_{\le t_0}} \le \prcond[g]{a_t^i \neq a_g \mid  s_{\le t_0}}
		\end{align*}
		Then~\eqref{eq:deviations-unlikely-Ht} follows using again Item~(iii) in \Cref{step:correlated1}.
	\end{step}

		We continue by setting up for the rest of the proof.
		For each $i \in N$ and each $t > t_0$, we say that agent $i$ defects in period $t$ if play follows $H_{<t}$ up to period $t-1$ and $i$ does not choose $a_g$ in period $t$, and we define
		\begin{align*}
			D_t^i = \left\{s^i \in S^T\colon \sigma_t^i(s_{\le t}^i; H_{<t}) \neq a_g\right\}
		\end{align*}
		as those infinite trajectories $s^i$ for $i$ such that $i$ defects at $t$ for $s_{\le t}^i$.
		Note that, by~\eqref{eq:deviations-unlikely-Ht}, for almost every $s_{\le t_0} \in\mathcal S_{\le t_0}$,
		\begin{align}
			\prcond[g]{D_t^i \mid  s_{\le t_0}, H_{<t}} = \prcond[g]{a_t^i \neq a_g \mid  s_{\le t_0}, H_{<t} }  \le 2 e^{-(r + 4\epsilon)t}\label{eq:defection-probability-bound-conditional}
		\end{align}
		
The rest of the proof proceeds as follows. 
To simplify the outline, we pretend that $t_0 = 0$ (and thus $\mathcal S_{\le t_0} = \{\emptyset\}$ contains only the empty signal profile trajectory) and that $\mathcal W = S^{N\times T}$. 
The goal is to show that $\prcond[f]{H} > 0$, which would imply that each agent's learning rate is $0$ and thus contradict the assumption.

Here is an initial attempt that fails but is a useful starting point. 
It would suffice to show that for each $i$ and each $t > t_0$, $\prcond[f]{D_t^i} \le e^{-\epsilon t}$, which would follow from $\prcond[g]{D_t^i} \le e^{-(r + 3\epsilon)t}$ by the definition of $r$.
This is close to~\eqref{eq:defection-probability-bound-conditional}, except that~\eqref{eq:defection-probability-bound-conditional} involves conditioning on $H_{<t}$. 
Unfortunately, $\prcond[g]{D_t^i} \le e^{-(r + 3\epsilon)t}$ does not follow from~\eqref{eq:defection-probability-bound-conditional} since $\prcond[g]{D_t^i \mid H_{<t}}$ can be arbitrarily small compared to $\prcond[g]{D_t^i}$.
This dead end inspires a more nuanced approach. 
 
We split up agent $i$'s infinite signal trajectories.
For each $t$, $E_t^i$ consists of those trajectories $s^i$ for which there is some $j$ such that $\prcond[g]{D_t^j \mid s^i, H_{<t}}$ is not too small and $t$ is the first period for which there is such $j$, and $F^i$ consists of those trajectories $s^i$ for which $\prcond[g]{D_t^j \mid s^i, H_{<t}}$ is small for all $j$ and $t$.
For a trajectory $s^i \in E_t^i \cup F^i$, $\prcond[g]{H_{<t}\mid s^i}$ is large since no agent is likely to defect in any period before $t$ conditional on $s^i$ (see~\eqref{eq:correlated-probHt} below).
Thus, one can bound the probability of $E_t^i$ in terms of its probability conditional on $H_{<t}$, which in turn can be bounded in terms of the probability that some agent defects in period $t$ conditional on a trajectory in $E_t^i$ and $H_{<t}$ by the definition of $E_t^i$ (see~\eqref{eq:correlated2} below).
Similarly, for each $t$, one can bound the probability of those trajectories in $F^i$ for which agent $i$ defects in period $t$ in terms of the same probability conditional on $H_{<t}$ (see~\eqref{eq:correlated5} below).
Combining both estimates and using that defections are very unlikely in state $g$ by assumption, gives that the intersection of $D_t^i$ with each of the sets $E_{t'}^i$, $t' \in T$, and $F^i$ has very low probability in state $g$ and thus still moderately low probability in state $f$.
But then, it is unlikely that any agent defects in state $f$, and all agents indefinitely play $a_g$ with positive probability. 
	
	\begin{step}\label{step:correlated5}
		Fourth, since signals are conditionally independent across periods, for each $i,j \in N$ and $t > t_0$,
		\begin{align}
			\prcond[g]{a_{t}^j \neq a_g \mid \mathfrak s_{\le t_0}, \mathfrak s^i, H_{<t}} = \prcond[g]{a_{t}^j \neq a_g \mid \mathfrak s_{\le t_0}, \mathfrak s_{\le t}^i, H_{<t}}
			\label{eq:deconditioning}
		\end{align}	
		almost surely.	
		For each $i\in N$ and each $t > t_0$, define inductively 
		\begin{align*}
			E_{t}^{i}(s_{\le t_0}) = & \left\{s^i \in \{s_{\le t_0}^i\}\times S^{\{t_0+1,t_0+2,\dots\}}\colon \exists j \in N, \prcond[g]{a_{t}^j \neq a_g \mid  s_{\le t_0}, s^i, H_{<t}} \ge e^{-\epsilon t}\right\}\\	 
			&\setminus \bigcup_{t > t' > t_0} E_{t'}^i(s_{\le t_0})
		\end{align*}
		and define
		\begin{align*}
			F^{i}(s_{\le t_0}) = \left\{s^i \in \{s_{\le t_0}^i\}\times S^{\{t_0+1,t_0+2,\dots\}} \colon \forall j\in N, \forall t > t_0, \prcond[g]{a_{t}^j \neq a_g \mid  s_{\le t_0}, s^i, H_{<t}} < e^{-\epsilon t}\right\}	
		\end{align*}
		and note that $\{E_t^i(s_{\le t_0}^i) \colon t > t_0\}\cup\{F^i(s_{\le t_0})\}$ is a partition of $\{s_{\le t_0}^i\} \times S^{\{t_0+1,t_0+2,\dots\}}$.
		Intuitively, $E_{t}^{i}(s_{\le t_0})$ is the set of infinite trajectories of $i$'s signals with prefix $s_{\le t_0}^i$ such that conditional on $s_{\le t_0}$, $s^i$, and $H_{<t}$, some agent defects with appreciable probability in period $t$ and no agent does so in any period between $t_0$ and $t$, and $F^{i}(s_{\le t_0})$ contains those infinite trajectories for which no agent defects with appreciable probability in any period after $t_0$.
		It may help to think of the sets $E_t^i(s_{\le t_0})$ as ``empty'' in the sense that some other agent exhausts an appreciable fraction of them through defections, and to think of the sets $F^i(s_{\le t_0})$ as ``full'' in the sense that they are not significantly reduced by a defection of any agent.
		
		We establish bounds on the probabilities of these sets. 
		First, for each $i\in N$, each $t > t_0$, almost every $s_{\le t_0}\in\mathcal S_{\le t_0}$, and almost every $s^i \in E_t^i(s_{\le t_0}) \cup F^i(s_{\le t_0})$,
		\begin{align}
			\begin{aligned}
				\prcond[g]{H_{<t} \mid  s_{\le t_0}, s^i} 
				&= \prod_{t > t' > t_0} \prcond[g]{H_{<t' + 1} \mid  s_{\le t_0}, s^i, H_{<t'}} \\
				&= \prod_{t > t' > t_0} \left(1 - \prcond[g]{\exists j \in N, a_{t'}^j \neq a_g \mid  s_{\le t_0}, s^i, H_{<t'}}\right) \\
				&\ge 1 - \sum_{j\in N,t > t' > t_0} \prcond[g]{a_{t'}^j \neq a_g \mid  s_{\le t_0}, s^i, H_{<t'}} \\
				&\ge 1 - n\sum_{t > t' > t_0} e^{-\epsilon t'} \\
				&\ge \frac12
			\end{aligned}
			\label{eq:correlated-probHt}
		\end{align}
		where the fourth step uses that $s^i\not\in E_{t'}^i(s_{\le t_0})$ for $t' < t$.
		Thus, for each $i\in N$, each $t > t_0$, and almost every $s_{\le t_0}$,
		\begin{align}
			\begin{aligned}
			\prcond[g]{E_t^i(s_{\le t_0}) \mid  s_{\le t_0}}
			&= \int_{E_t^i(s_{\le t_0})}  d\prcond[g]{s^i \mid s_{\le t_0}}\\
			&\le 2\int_{E_t^i(s_{\le t_0})} \prcond[g]{H_{<t} \mid s_{\le t_0},s^i} d\prcond[g]{s^i \mid s_{\le t_0}}\\
			&= 2\prcond[g]{E_t^i(s_{\le t_0})\cap H_{<t} \mid s_{\le t_0}}\\
			&\le 2\prcond[g]{E_t^i(s_{\le t_0}) \mid  s_{\le t_0}, H_{<t}} \\
			&= 2 \int_{E_t^i(s_{\le t_0})} d\prcond[g]{s^i \mid  s_{\le t_0}, H_{<t}} \\
			&\le 2e^{\epsilon t} \int_{E_t^i(s_{\le t_0})} \sum_{j \in N}\prcond[g]{a_t^j \neq a_g \mid  s_{\le t_0}, s^i, H_{<t}} d\prcond[g]{s^i \mid  s_{\le t_0}, H_{<t}} \\
			&\le 2e^{\epsilon t} \sum_{j \in N} \prcond[g]{a_{t}^j \neq a_g\mid  s_{\le t_0}, H_{<t}} \\
				&\le 4n e^{\epsilon t} e^{-(r + 4\epsilon)t}\\ 
				&= 4n e^{-(r + 3\epsilon)t}
			\end{aligned}
			\label{eq:correlated2}
		\end{align}
		where the second step uses~\eqref{eq:correlated-probHt}, the sixth step uses the definition of $E_t^i(s_{\le t_0})$, and the second to last step uses~\eqref{eq:deviations-unlikely-Ht}.
		In words, the left-hand side is the probability that, conditional on $g$ and $s_{\le t_0}$, agent $i$ observes an infinite trajectory conditional on which some agent defects with appreciable probability in period $t$ and no agent defects with appreciable probability in any period before $t$. 	
		Second, for each $i\in N$, each $t > t_0$, and almost every $s_{\le t_0}\in\mathcal S_{\le t_0}$,		
		\begin{align}
			\begin{aligned}
			\prcond[g]{\left(D_t^i \cap F^i(s_{\le t_0})\right)_{\le t} \mid  s_{\le t_0}}
			&\le 2\prcond[g]{\left(D_t^i \cap F^i(s_{\le t_0})\right)_{\le t} \mid  s_{\le t_0}, H_{<t}} \\
			&\le 2\prcond[g]{a_t^i \neq a_g \mid  s_{\le t_0}, H_{<t}}\\
			&\le 4 e^{-(r + 4\epsilon)t}
			\end{aligned}
			\label{eq:correlated5}
		\end{align}
		where the first step follows from repeating the first four steps in~\eqref{eq:correlated2} and the third step follows from~\eqref{eq:deviations-unlikely-Ht}.
		Here, the left-hand side is the probability that conditional on $g$ and $s_{\le t_0}$, agent $i$ observes a trajectory up to $t$ for which she defects in period $t$ and, conditional on which no agent defects with appreciable probability in any period from $t_0+1$ to $t$. 	
	\end{step}
	\begin{step}\label{step:correlated6}
		We now show that, conditional on state $f$ and almost every signal profile trajectory $s_{\le t_0}\in\mathcal S_{\le t_0}$, with positive probability, no agent defects in any period.
		
		First, for each $i \in N$, each $t > t_0$, and almost every $s_{\le t_0} \in\mathcal S_{\le t_0}$,
		\begin{align*}
			\begin{aligned}
				\prcond[f]{ E_t^i(s_{\le t_0}) \mid  s_{\le t_0}, \mathcal W} &\le 2 \prcond[f]{ E_t^i(s_{\le t_0}) \cap \mathcal W \mid  s_{\le t_0}}\\
				&\le 2 \prcond[f]{ (E_t^i(s_{\le t_0}) \cap \mathcal W^i)_{\le t} \mid  s_{\le t_0}}\\
				&\le 2 e^{(m^i+\epsilon)t_0} e^{-(m^i-\epsilon)t_0 + (\evcond[f]{\ell_{f,g}^i}+\epsilon)(t-t_0)} \prcond[g]{ (E_t^i(s_{\le t_0}) \cap \mathcal W^i)_{\le t} \mid  s_{\le t_0}} \\
				&\le 2 e^{2\epsilon t_0 + (\evcond[f]{\ell_{f,g}^i}+\epsilon)(t-t_0)} \prcond[g]{ (E_t^i(s_{\le t_0}) \cap \mathcal W^i)_{\le t} \mid  s_{\le t_0}} \\
				&\le 2 e^{(r+2\epsilon)t} \prcond[g]{ (E_t^i(s_{\le t_0}) \cap \mathcal W^i)_{\le t} \mid  s_{\le t_0}} \\
				&\le 8 ne^{-\epsilon t}	
			\end{aligned}
		\end{align*}
		where the first step uses that $\prcond[f]{\mathcal W \mid s_{\le t_0}} \ge \frac{1}{2}$ by \Cref{step:correlated2}, the third step uses~\eqref{eq:correlated-likelihood-increase}, and the last step uses~\eqref{eq:correlated2} and the fact that, by~\eqref{eq:deconditioning}, whether an infinite trajectory of $i$'s signals is in $E_t^i(s_{\le t_0})$ only depends on its prefix up to and including $t$.
		Thus,
		\begin{align}
			\begin{aligned}
				\prcond[f]{\cup_{i\in N, t > t_0} E_t^i(s_{\le t_0}) \mid  s_{\le t_0}, \mathcal W} 
				\le 8n^2 \sum_{t > t_0} e^{-\epsilon t}
				\le \frac14
			\end{aligned}
			\label{eq:correlated7}
		\end{align}
		Second, similar to above, using~\eqref{eq:correlated5} instead of~\eqref{eq:correlated2}, for each $i\in N$, each $t > t_0$, and almost every $s_{\le t_0} \in \mathcal S_{\le t_0}$,
		\begin{align*}
			\prcond[f]{D_t^i \cap F^{i}(s_{\le t_0}) \mid  s_{\le t_0}, \mathcal W}
			&\le 2\prcond[f]{D_t^i \cap F^{i}(s_{\le t_0}) \cap \mathcal W \mid  s_{\le t_0}}\\ 
			&\le 2\prcond[f]{\left(D_t^i \cap F^{i}(s_{\le t_0})_{\le t} \cap \mathcal W^i\right)_{\le t} \mid  s_{\le t_0}}\\
			&\le 2e^{(r+2\epsilon)t} \prcond[g]{\left(D_t^i \cap F^{i}(s_{\le t_0}) \cap \mathcal W^i\right)_{\le t} \mid  s_{\le t_0}} \\
			&\le 8e^{-2\epsilon t}
		\end{align*}
		and thus,
		\begin{align}
			\begin{aligned}
				\prcond[f]{\cup_{i\in N, t > t_0} D_t^i \cap F^i(s_{\le t_0}) \mid  s_{\le t_0}, \mathcal W}	
			\le 4n\sum_{t > t_0} e^{-2\epsilon t}
			\le \frac14
			\end{aligned}
			\label{eq:correlated9}
		\end{align}
		Then, for almost every $s_{\le t_0}\in\mathcal S_{\le t_0}$,
		\begin{align*}
			1 - \prcond[f]{H \mid  s_{\le t_0}, \mathcal W} 
			&\le \prcond[f]{\cup_{i\in N, t > t_0} D_t^i \mid  s_{\le t_0}, \mathcal W}\\
			&= \prcond[f]{\cup_{i\in N, t,t' > t_0} D_t^i  \cap E_{t'}^i(s_{\le t_0}) \mid  s_{\le t_0}, \mathcal W}  \\
			&+ \prcond[f]{\cup_{i\in N, t > t_0} D_t^i \cap F^i(s_{\le t_0}) \mid  s_{\le t_0}, \mathcal W}  \\
			&\le \prcond[f]{\cup_{i\in N, t > t_0} E_{t}^i(s_{\le t_0}) \mid  s_{\le t_0}, \mathcal W}  \\
			&+ \prcond[f]{\cup_{i\in N, t > t_0} D_t^i \cap F^i(s_{\le t_0}) \mid  s_{\le t_0}, \mathcal W}  \\
			&\le\frac1{2}
		\end{align*}
		If the infinite history is not $H$, some agent defects in some period, hence the first step. 
		The second step uses that $\{E_{t}^i(s_{\le t_0})\colon t > t_0\} \cup \{F^i(s_{\le t_0})\}$ is a partition of $\{s_{\le t_0}^i\}\times S^{\{t_0+1,t_0+2,\dots\}}$.
		The third step is a basic manipulation.
		The fourth step follows from~\eqref{eq:correlated7} and~\eqref{eq:correlated9}.
		Hence, \Cref{step:correlated2}, $\prcond[g]{\mathcal S_{\le t_0}} > 0$, and \Cref{lem:absolute-continuity} imply that
		\begin{align*}
			\prcond[f]{H} &\ge \int_{S_{\le t_0}} \prcond[f]{H \mid  s_{\le t_0}, \mathcal W} \prcond[f]{\mathcal W\mid s_{\le t_0}} d\prcond[f]{s_{\le t_0}}  \ge \frac12 \int_{\mathcal S_{\le t_0}} \prcond[f]{\mathcal W\mid s_{\le t_0}} d\prcond[f]{s_{\le t_0}} > 0
		\end{align*}
		In particular, in state $f$, the probability that each agent chooses $a_g$ in each period after $t_0$ is strictly positive, and thus each agent's learning rate is $0$.
		This contradicts~\eqref{eq:too-fast-learning} and finishes the proof.
	\end{step}
\end{proof}

We prove that sufficiently many agents in a strongly connected network can learn faster than a single agent in autarky.

\learninglowerbound*
\begin{proof}
	We assume for now that $N^i = N$ for all $i\in N$ and treat the general case later.
	For $f,g\in\Omega$, define $\ell_{f,g} = \ell_{f,g}^1$ and $m_{f,g} = \evcond[f]{\ell_{f,g}}$, and fix $\epsilon > 0$.
	First, by~\eqref{eq:lambda-star-state-swap}, $\lambda_{f,g}^*(-m_{g,f}) = m_{g,f}$.
	Second, since $I_{f,g}^*$ contains $m_{g,f}$ and $\lambda_{f,g}^*$ is continuous on $I_{f,g}^*$ by \Cref{lem:lambda-star-properties} and the remarks thereafter, there is $\delta > 0$ such that 
	\begin{align*}
		\min_{f\neq g} \lambda_{f,g}^*(-m_{g,f} + \delta) > \min_{f\neq g} \lambda_{f,g}^*(-m_{g,f}) - \epsilon
	\end{align*}
	We may further assume that $\delta < \min_{f\neq g} m_{f,g}$.
	Hence, when defining
	\begin{align*}
		\rmajt = \min_{f\neq g} \lambda_{f,g}^*(-m_{g,f} + \delta)
	\end{align*}
	we have $\rmajt > \min_{f\neq g} m_{f,g} = \rbdd - \epsilon$.
	It thus suffices to exhibit strategies for which the learning rate is at least $\rmajt$.
	\setcounter{step}{0}
	\begin{step}[Constructing the strategies]\label{step:fast-learning1}
	We start by inductively defining the strategies.
	Let $i\in N$.
	The strategy $\sigma_1^i$ is arbitrary.
	Now let $t > 1$. 
	Given $a^1_{t-1},\dots,a^n_{t-1}$, let $a_{t-1}^{\mathrm{pop}} \in \arg\max_{a\in A} |\{j\in N\colon a_{t-1}^j = a\}|$ be an action that is most popular among the actions taken in period $t-1$.
	For a history $H_{< t}^i \in A^{N^i\times (t-1)}$ for agent $i$, let
	\begin{align*}
		a_t^i = \sigma_t^i(\mathfrak s_1^i,\dots,\mathfrak s_t^i;H_{< t}^i) =
		\begin{cases}
			a_f \quad&\text{if } L_{f,g,t}^{i} \ge (m_{f,g} - \delta) t \;\forall g\neq f\text{, and}\\
			a_{t-1}^{\mathrm{pop}} \quad&\text{otherwise}
		\end{cases}
	\end{align*}
	Note that $\sigma^i$ is well-defined since $L_{f,g,t}^i = -L_{g,f,t}^i$ and $m_{f,g} - \delta > 0$.
	Hence, agents follow their private signals if those are sufficiently decisive and the previous period's most popular action otherwise. 
	\end{step}

	\begin{step}[Bounding the probabilities of mistakes]\label{step:fast-learning2}
		 Now we derive the claimed bound on the probability of mistakes.
		 
		 \setcounter{case}{0}
		 \begin{case}
		 	First, we consider the probability that agent $i$ makes a mistake if she acts based on her private signals.
		 By \Cref{thm:large-deviations} and the remarks after \Cref{lem:lambda-star-properties}, we have for any two distinct $f,g\in\Omega$,
		 \begin{align*}
		 	\prcond[g]{L_{g,f,t}^{i} \le -(m_{f,g} - \delta) t} = e^{-\lambda_{g,f}^*(-m_{f,g} + \delta) t + o(t)} 
		 \end{align*}
		Hence, for each $f\in\Omega$,
		 \begin{align*}
		 	\pr[a_t^i \neq a_\omega \mid L_{f,g,t}^{i} \geq (m_{f,g} - \delta) t\;\forall g\neq f] &\le \sum_{g\neq f} \prcond[g]{L_{f,g,t}^{i} \ge (m_{f,g} - \delta) t} \pr[\omega = g]\\	
			&= \sum_{g\neq f} \prcond[g]{L_{g,f,t}^{i} \le -(m_{f,g} - \delta) t} \pr[\omega = g]\\
		 	&= \sum_{g\neq f} e^{-\lambda_{g,f}^*(-m_{f,g} + \delta)t + o(t)} \pr[\omega = g]\\
		 	&= e^{-\rmajt t + o(t)}
		 \end{align*}	
		which gives the desired bound.
		\end{case}
		\begin{case}
			Second, we consider the probability that agent $i$ makes a mistake if she follows the previous period's most popular action.
			We use a standard tail estimate for binomial distributions: if $X$ is binomially distributed with sample size $n$ and success probability $q$, then
			\begin{align*}
				\pr[X \le k] \le e^{-n \infdiv{\frac kn}{q}}
			\end{align*}
			where 
			\begin{align*}
			\infdiv{a}{b} = a\log \frac ab + (1-a)\log \frac{1-a}{1-b}
		\end{align*}
		is the Kullback-Leibler divergence of two Bernoulli distributions with success probabilities $a,b\in [0,1]$.
		Thus, for all $f\in\Omega$, 
		\begin{align*}
			\prcond[f]{a_{t-1}^{\mathrm{pop}} \neq a_f} &\le \prcond[f]{|\{i\in N\colon L_{f,g,t-1}^{i} \geq (m_{f,g} - \delta)  (t-1)\; \forall g\neq f\}| \le \frac n2}\\
			&\le e^{-n \infdiv{\frac12}{1-q_{f,t-1}}}
		\end{align*}
		where
		\begin{align*}
			q_{f,t} = \sum_{g\neq f} e^{-\lambda_{f,g}^*(m_{f,g} - \delta)t + o(t)} \pr[\omega = g] = e^{- \min_{f\neq g} \lambda_{f,g}^*(m_{f,g} - \delta) t + o(t)}
		\end{align*}
		is an upper bound for the probability that agent $i$ does not choose $a_f$ in state $f$ (obtained from \Cref{thm:large-deviations}).
		Hence, for all $f\in\Omega$,
		\begin{align*}
			\infdiv{\frac12}{1-q_{f,t}} &= \frac12\left(\log\frac1{2(1-q_{f,t})} + \log \frac1{2q_{f,t}}\right)\\
			&= \log \frac12 + \frac12\left(\log\frac1{1-q_{f,t}} + \log \frac1{q_{f,t}}\right)\\
			&= \log\frac12 + \frac12\left(q_{f,t} + o(q_{f,t}) + \min_{f\neq g} \lambda_{f,g}^*(m_{f,g} - \delta)t + o(t)\right)\\
			&= \frac12 \min_{f\neq g} \lambda_{f,g}^*(m_{f,g} - \delta) t + o(t)
		\end{align*}
		where the third step uses that $\frac{1}{1-x} = 1 + x + o(x)$ and $\log(1+x + o(x)) = x + o(x)$, where $o(x)$ is a quantity that converges to $0$ faster than $x$ as $x\rightarrow 0$.
		Thus,
		\begin{align*}
			\prcond[f]{a_{t-1}^{\mathrm{pop}} \neq a_f} &\le e^{-\frac n2 \min_{f\neq g} \lambda_{f,g}^*(m_{f,g} - \delta) t + o(t)}
		\end{align*}
		Since $\lambda_{f,g}^*(m_{f,g}) = 0$ and $\lambda_{f,g}^*$ is strictly decreasing on $(\inf \ell_{f,g},m_{f,g}]$ by \Cref{lem:lambda-star-properties}\ref{item:lambda-star-properties-monotonicity}, $\lambda_{f,g}^*(m_{f,g} - \delta) > 0$ for all $f\neq g$, and so $\min_{f\neq g} \lambda_{f,g}^*(m_{f,g} - \delta) > 0$.	
		So if $n$ is large enough depending on $\delta$, the probability that $i$ makes a mistake in period~$t$ by following the most popular action of period $t-1$ decreases at a rate of at least $\rmajt$.
		More precisely, we need that
		\begin{align*}
			n \ge 2 \frac{\min_{f\neq g} \lambda_{f,g}^*(-m_{g,f} + \delta)}{\min_{f\neq g} \lambda_{f,g}^*(m_{f,g} - \delta)}
		\end{align*}
		\end{case}
		We conclude that if $n$ is large enough depending on $\delta$, then each agent learns at a rate of at least $\rmajt = \min_{f\neq g} \lambda_{f,g}^*(-m_{g,f} + \delta)$.

	\end{step}
	\begin{step}[Extending to arbitrary networks]\label{step:fast-learning3}
		It remains to extend the results to arbitrary strongly connected networks.
		We sketch the argument but omit the details. 
		The main idea is to add periods that are used to propagate the agents' action choices in previous periods through the network.
		
		For two agents $i,j$, denote by $d(i,j)$ the length of a shortest path from $i$ to $j$.
		For example, if $j \in N^i$ and $i\neq j$, then $d(i,j) = 1$.\footnote{More precisely, $d(i,j)$ is defined inductively by letting $d(i,i) = 0$ and for $k\ge 1$ and $j\neq i$, $d(i,j) = k$ if $\min\{d(i,i')\colon i'\in N, d(i,i') \le k-1, j \in N^{i'}\} = k-1$.}
		Since the network is strongly connected, $d(i,j)$ is at most $n-1$ for all $i,j$.
		We partition the set of periods $T$ into intervals of $M = 1 + n(n-2)$ periods.
		We call the periods $\{1,1+M,1+2M,\dots\}$ voting periods, and the remaining periods propagation periods.
		In each voting period $t\in\{1,1+M,1+2M,\dots\}$, each agent $i$ chooses an action similar to the construction in \Cref{step:fast-learning1}: if $t = 1$, $i$ chooses an arbitrary action; in any later voting period $t$, if $i$'s private signals up to period~$t$ are sufficiently decisive, she chooses an action optimally based on those, and she follows the most popular action in period $t-M$ otherwise.
		(The strategies during the propagation periods will ensure that $i$ knows the most popular action in period $t-M$ even though she does not observe all agents' actions directly.)
		For $j\in N$, in period $t \in\{1+j, 1+j + M, 1+j+2M,\dots\}$, each agent $i$ with $d(i,j) = 1$ imitates $j$'s action in period $t-j$ (i.e., $a_t^i = a_{t-j}^j$), and all other agents repeat their own action in period $t-j$ (i.e., $a_t^i = a_{t-j}^i$).
		For $j\in N$ and $k\in[n-3]$, in period $t\in\{1 + j + kn, 1+j+kn + M,1+j+kn + 2M, \dots\}$, each agent $i$ with $d(i,j) = k+1$ imitates $j$'s action in period $t - j - kn$ (which they observed from some agent with distance $k$ to $j$ in period $t - n$), and all other agents repeat their own action in period $t - j - kn$.
		Hence, any propagation period~$t$ with $t\in\{1 + j + kn, 1+j+kn + M,1+j+kn + 2M, \dots\}$ is used to inform agents with distance $k+2$ to $j$ about $j$'s action in the latest voting period by letting an agent with distance $k+1$ to $j$ imitate that action.\footnote{Extending \Cref{thm:coordination} to random networks as discussed in \Cref{rem:random-networks} requires a minor modification of the strategies: in each period $t\in\{1+j+kn,1+j+kn + M,1+j+kn+2M,\dots\}$, each agent who has observed $j$'s action in the latest voting period either directly or indirectly through other agents, imitates $j$'s action in that voting period, and all other agents repeat their own action in that voting period. Then, the set of agents who have (possibly indirectly) observed $j$'s action in the latest voting period grows by at least one agent in period~$t$, unless it already contained all agents before period~$t$. This follows from the assumption that any realization of the network is strongly connected, so that in period~$t$, at least one new agent observes the action of some agent already contained in that set before period~$t$.\label{fn:random-networks}}
		
		Since the agents know the network, they know whether the action of an agent $i$ in any propagation period imitates the action of another agent or agent $i$'s own action in the latest voting period.
		Thus, in each voting period, each agent knows all other agents' actions in the preceding voting period.
		By the same arguments as in \Cref{step:fast-learning2}, for each $i\in N$ and each voting period $t\in\{1,1+M,1+2M,\dots\}$,
		\begin{align*}
			\pr[a_t^i\neq a_\omega] \le e^{-\rmajt t + o(t)}
		\end{align*}
		provided that $n$ is large enough.
		In any propagation period~$t$, each agent $i$ imitates the action of an agent from the latest voting period, and so
		\begin{align*}
			\pr[a_t^i\neq a_\omega] \le e^{-\rmajt (t-M) + o(t)} = e^{-\rmajt t + o(t)}
		\end{align*}
		since that voting period does not lie more than $M$ periods in the past.
		So the preceding inequality holds for all periods after replacing the $o(t)$-term.
	\end{step}
\end{proof}

\newpage

\section{Notation}\label{sec:referencetable}
	 
The table below summarizes the notation used in the paper.	 
	 
\vspace{2ex}

\centering
\begin{tabularx}{\textwidth}{llXl}
	\toprule
	\multicolumn{2}{l}{Symbol [elements]} & Name & Mathematical object\\
	\midrule
	$N$ & $i,j$ & agents & $\{1,\dots,n\}$ \\
	$T$ & $t$ & periods & $\{1,2,\dots\}$ \\
	$A$ & $a,a_t^i$ & actions & arbitrary set\\
	$\Omega$ & $f,g$ & states & finite set\\
	$\omega$ & & random unknown state & $\Omega$-valued random variable \\
	$\pi_0$ & & common prior belief & distribution on $\Omega$\\
	$u$ & & utility function & function $A\times\Omega\rightarrow\mathbb R$\\
	$a_f$ & & utility-maximizing action in state $f$ & element of $A$\\
	$S$ & $s,s_t^i$ & signals & standard Borel space\\
	$\mathfrak s_t^i$ & & $i$'s private signal in period $t$ & $S$-valued random variable \\
	$\mu_f^i$ & & distribution of $\mathfrak s_t^i$ in state $f$ & distribution on $S$ \\
	$\ell_{f,g}^i(s)$ & & $i$'s log-likelihood ratio given signal $s$ & function $S\rightarrow \mathbb R$ \\
	$\lambda_{f,g}^i(z)$ & & CGF of $\ell_{f,g}^i(\mathfrak s_1^i)$ & function $\mathbb R\rightarrow \mathbb R$ \\
	$(\lambda_{f,g}^i)^*(\eta)$ & & Fenchel-Legendre transform of $\lambda_{f,g}^i(z)$ & function $\mathbb R\rightarrow \mathbb R$ \\
	$L_{f,g,t}^i$ & & $i$'s log-likelihood ratio given $\mathfrak s_{\le t}^i$ & $\mathbb R$-valued random variable \\
	$N^i$ & $j$ & $i$'s neighborhood & subset of $N$ \\
	$\mathcal I_{\le t}^i$ &  & $i$'s information sets in period $t$ & short for $S^t\times A^{N^i\times(t-1)}$ \\ 
	$\sigma_t^i$ & & $i$'s strategy in period $t$ & function $\mathcal I_{\le t}^i\rightarrow A$ \\
	$\sigma$ & & strategy profile & short for $(\sigma_t^i)_{i\in N,t\in T}$ \\
	$\mathcal S^i(H_{\le t})$ & $s_{\le t}^i$ & $i$'s signals compatible with $H_{\le t}$ & subset of $S^t$ \\
	$\mathcal S(H_{\le t})$ & $s_{\le t}$ & signal profiles compatible with $H_{\le t}$ & short for $\Pi_{i\in N} \mathcal S^i(H_{\le t})$ \\
	$\raut^i$ & & $i$'s learning rate in autarky & element of $\mathbb R_+$ \\
	$\rbdd$ & & tight learning rate upper bound & element of $\mathbb R_+$ \\
	$\rbddt$ & & relaxed learning rate upper bound & element of $\mathbb R_+$ \\
	\bottomrule
\end{tabularx}

\end{document}